\documentclass[a4paper,UKenglish]{lipics-v2018}

\usepackage{microtype}

\usepackage{amsfonts,amsmath,amsthm}
\usepackage{graphicx}
\usepackage[noend]{algpseudocode}
\usepackage{algorithm}
\usepackage{comment}
\usepackage[hyperref]{xcolor}
\hypersetup{
    colorlinks,
    linkcolor={red!50!black},
    citecolor={blue!100!black},
    urlcolor={blue!80!black}
}

\usepackage{multicol}
\usepackage{xcolor}
\usepackage{comment}
\usepackage{multirow}
\usepackage{xspace}

\newtheorem{claim}{Claim}

\newcommand{\graphval}[1]{\ensuremath{\mbox{GraphVal}_{#1}}}

\newcommand{\eqmcost}[1]{\ensuremath{\mbox{EqMCost}_{#1}}}
\newcommand{\eqmcostty}[1]{\ensuremath{\mbox{EqMCostTy}_{#1}}}
\newcommand{\mcG}{\mathcal{G}}
\newcommand{\mcT}{\mathcal{T}}
\newcommand{\low}{\ensuremath{\mbox{Low}}}
\newcommand{\high}{\ensuremath{\mbox{High}}}
\newcommand{\med}{\ensuremath{\mbox{Mid}}}
\newcommand{\graphflow}{\ensuremath{\mbox{GraphFlow}}\xspace}
\newcommand{\redistribh}{\ensuremath{\widehat{\mbox{Redistrib}}}\xspace}
\newcommand{\graphflowh}{\ensuremath{\widehat{\mbox{GraphFlow}}}\xspace}
\newcommand{\binsearchh}{\ensuremath{\widehat{\mbox{BinSearch}}}\xspace}
\newcommand{\graphvalh}[1]{\ensuremath{\widehat{\mbox{GraphVal}_{#1}}}}
\newcommand{\eqmcosth}[1]{\ensuremath{\widehat{\mbox{EqMCost}_{#1}}}}
\newcommand{\nphard}{NP-hard\xspace}
\newcommand{\subsum}{\textsf{SUBSET-SUM}\xspace}

\algrenewcommand\algorithmicrequire{\textbf{Input:}}
\algrenewcommand\algorithmicensure{\textbf{Output:}}

\title{Equilibrium Computation in Atomic Splittable Routing Games with Convex Cost Functions}

\author{Umang Bhaskar}{Tata Institute of Fundamental Research, Mumbai, India}{umang@tifr.res.in}{}{}

\author{Phani Raj Lolakapuri}{Tata Institute of Fundamental Research, Mumbai, India}{phaniraj@tcs.tifr.res.in}{}{}
\authorrunning{U. Bhaskar and P.\,R. Lolakapuri}

\Copyright{Umang Bhaskar and Phani Raj Lolakapuri}

\subjclass{Theory of computation $ \rightarrow $  Network games}

\category{}

\relatedversion{}

\supplement{}

\acknowledgements{ We thank Isha Tarte, who as an intern at TIFR, helped in verifying the results in Section~\ref{sec:hardness}. We also thank Varun Narayanan for suggesting the reduction to Gr\"{o}ebner basis to find the common zeros of the polynomials.}

\EventEditors{John Q. Open and Joan R. Access}
\EventNoEds{2}
\EventLongTitle{42nd Conference on Very Important Topics (CVIT 2016)}
\EventShortTitle{CVIT 2016}
\EventAcronym{CVIT}
\EventYear{2016}
\EventDate{December 24--27, 2016}
\EventLocation{Little Whinging, United Kingdom}
\EventLogo{}
\SeriesVolume{42}
\ArticleNo{23}
\nolinenumbers 
\hideLIPIcs  

\begin{document}

\maketitle

\begin{abstract}
We present polynomial-time algorithms as well as hardness results for equilibrium computation in atomic splittable routing games, for the case of general convex cost functions. These games model traffic in freight transportation, market oligopolies, data networks, and various other applications. An atomic splittable routing game is played on a network where the edges have traffic-dependent cost functions, and player strategies correspond to flows in the network. A player can thus split it's traffic arbitrarily among different paths. While many properties of equilibria in these games have been studied, efficient algorithms for equilibrium computation are known for only two cases: if cost functions are affine, or if players are symmetric. Neither of these conditions is met in most practical applications. We present two algorithms for routing games with general convex cost functions on parallel links. The first algorithm is exponential in the number of players, while the second is exponential in the number of edges; thus if either of these is small, we get a polynomial-time algorithm. These are the first algorithms for these games with convex cost functions. Lastly, we show that in general networks, given input $C$, it is NP-hard to decide if there exists an equilibrium where every player has cost at most $C$. 
 \end{abstract}
\section{Introduction}

The problem of equilibrium computation, particularly efficient computation, is the cornerstone of algorithmic game theory, and is an area where researchers have had many successes. In many games, we have a good understanding of where the boundaries of computation lie, including normal-form games~\cite{Daskalakis13}, markets~\cite{GargMVY17}, and congestion games~\cite{FabrikantPT04}. The study of equilibrium computation has had a significant impact on algorithms, contributing new techniques and complexity classes.

In this paper, we are interested in equilibrium computation in atomic splittable routing games (ASRGs) with convex cost functions. These games are used to model many applications, including freight transportation, market oligopolies, and data networks (e.g.,~\cite{cominetti,orda}). In an ASRG, we are given a network with cost functions on the edges, and $k$ players. Each player $i$ has a source $s_i$, destination $t_i$, and a fixed demand $v_i$. Each player needs to transport its demand from its source to its destination at minimum cost, and is free to split the demand along multiple paths. Each player thus computes a minimum cost $s_i$-$t_i$ flow, given the strategy of the other players.

The fact that each player can split its flow along multiple paths is what differentiates these from weighted congestion games. This freedom reduces the combinatorial structure of the game, making ASRGs harder to analyze. For example, equilibria in ASRGs may be irrational. While equilibrium computation in (unsplittable) congestion games is well-studied, much less is known about ASRGs. In fact, properties of ASRGs apart from equilibrium computation have been studied. We know tight bounds on the price of anarchy~\cite{cominetti,harks_poa,roughgarden2015local}, and can characterize games with multiple equilibria~\cite{umang}.

However for equilibrium computation little is known. We know of only two cases when an equilibrium can be efficiently computed --- when cost functions are affine, or when players are symmetric, i.e., they have the same source, destination, and demand~\cite{cominetti,HarksT17}. These conditions are hardly ever met in practice. We know of no hardness results for this problem. A number of iterative algorithms for equilibrium computation are proposed, and sufficient conditions for convergence are given by Marcotte~\cite{marcotte}. Further, it is implicit in a paper by Swamy that one can compute equilibrium efficiently, given the total flow on each edge~\cite{Swamy12}.

Computing equilibria in ASRGs is an interesting theoretical challenge as well. In some regards, properties of pure Nash equilibria in ASRGs resemble mixed Nash equilibria in games. For example, an equilibrium in pure strategies always exists~\cite{rosen}. For many games with this property, local search algorithms are known that converge to an equilibrium (e.g., congestion games). However in ASRGs we do not know of any such algorithms.

In this work, we focus on polynomial time algorithms for computing equilibria in ASRGs with general convex costs on parallel edges. 
Parallel edges are interesting because a number of applications can be modeled using parallel edges, such as load balancing across servers~\cite{suri2007selfish}, and in traffic models~\cite{harks2015computing}. Further, many results were first obtained for graphs consisting of parallel edges and then extended (e.g., results on the price of collusion~\cite{HayrapetyanTW06}, extended to series-parallel graphs~\cite{BhaskarFH10}, or on the price of anarchy~\cite{KoutsoupiasP09,BhawalkarGR14}). These are thus a natural starting point to  study equilibrium computation. We believe it likely that some of our structural results extend beyond parallel edges, to nearly-parallel and series-parallel graphs. 

\subparagraph*{Our Contribution.} For ASRGs with convex costs on parallel edges, we give two algorithms. Our first algorithm computes an equilibrium\footnote{We use the standard notion of polynomial-time computation when outputs are possibly irrational: we say an algorithm is efficient if for any $\epsilon > 0$, the algorithm computes and $\epsilon$-approximate solution in time polynomial in the inputs size and $\log (1/\epsilon)$.} in time $O\left( \left(\log |\mathcal{I}|\right)^n \right)$, where $|\mathcal{I}|$ is the input size and $n$ is the number of players. If the  number of players is near-logarithmic in the input size, i.e., $O(\log |\mathcal{I}|/\log \log |\mathcal{I}|)$, this gives a polynomial time algorithm. Our algorithm is based on the idea of reducing equilibrium computation to guessing the marginal costs of the players at equilibrium. The marginal costs turn out to have a number of interesting monotonicity properties, which we use to give a high-dimensional binary search algorithm.

Our second algorithm has running time exponential in the number of edges in the network. If the number of edges is constant, then this gives us a polynomial-time algorithm. Define players to be of the same \emph{type} if they have flow on the same set of edges at equilibrium. The algorithm is based on the following structural result: for parallel edges, computing equilibrium in a general ASRG can be reduced to computing equilibrium in an ASRG where players of the same type also have the same demand. At a high level, this allows us to replace players of the same type by a single player, and then use our previous algorithm. Somewhat surprisingly, this result does not subsume the previous result. This is because the actual partition of players into types is unknown, hence we must enumerate over all possible such partitions, which introduces a factor of $O(n^{|E|})$ to the running time.

Lastly, we show that in general networks, determining existence of a Nash equilibrium where the cost of every player is at most $C$ is \nphard. Our proof here is a reduction from \subsum, and builds upon a construction showing multiplicity of equilibria in ASRGs~\cite{umang}. Our result parallels early results for bimatrix games~\cite{gilboa}, which showed that it is \nphard to determine existence of a Nash equilibrium in bimatrix games where the cost of players is above a threshold~\cite{gilboa}. Our proof is computer-assisted, and we use Mathematica to verify properties of equilibria in the games used in our reduction.

\textbf{Related Work.} Existence of equilibria in atomic splittable routing games (and for more general \emph{concave games}) is shown by Rosen~\cite{rosen}. Equilibria in these games is unique when delay functions of the edges are polynomials of degree $\leq 3$~\cite{altman}, when the players are symmetric, or when the underlying network is nearly-parallel~\cite{umang}. In general, the equilibria may not be unique~\cite{umang}. For computation, the equilibria can be obtained as the solution to a convex problem if the edge costs are linear, or if the players are symmetric~\cite{cominetti}. Huang considers ASRGs with linear delays on a class of networks called \emph{well-designed} which includes series-parallel graphs, and gives a combinatorial algorithm to find the equilibrium~\cite{Huang13}. A network is well-designed if  for the optimal flow (which minimizes total cost), increasing the total flow value does not decrease the flow on any edge. Recently, Harks and Timmermans give an algorithm to compute equilibrium for ASRGs with player specific-linear costs on parallel links~\cite{HarksT17}. This setting allows players to have different cost functions on an edge. Their results use a reduction to integrally splittable flows, where the flow each player puts on an edge is an integral multiple of some quantity. In our case, the equilibrium flow can be irrational, hence these ideas do not seem to work. A number of algorithms for equilibrium computation are also proposed by Marcotte, who shows convergence results for these~\cite{marcotte}.

Nonatomic games have an infinite set of players, each of which has infinitesimal flow. Unlike ASRGs, equilibria in nonatomic games are well-studied: an equilibrium can be obtained by solving a convex program, and is unique if the costs are strictly nondecreasing~\cite{beckmann}. It is also known that ASRGs captures the setting where nonatomic players to form coalitions, and within a coalition players cooperate to minimize the total cost of its flow~\cite{hayra}. A number of papers study the change in total cost as players in a nonatomic game form coalitions, forming an ASRG~\cite{hayra,BhaskarFH10,Huang13}. Another property of ASRGs that has received a lot of attention is the price of anarchy (PoA), formalised as the ratio of the total cost of the worst equilibrium, to the optimal cost. Upper bounds on the PoA were obtained by Cominetti, Correa, and Stier-Moses~\cite{cominetti} and improved upon by Harks~\cite{harks_poa}. These bounds were shown to be tight~\cite{roughgarden2015local}. 

Atomic games where demands are unsplittable are also extensively studied. If all players have the same demand, these are called congestion games. Player strategies may not correspond to paths in a graph; if they do, these are called network congestion games. For congestion games, existence of a potential function is well-known~\cite{rosenthal}, though computing an equilibrium is PLS-hard~\cite{FabrikantPT04,AckermannRV06}, even if players are symmetric or the edge cost functions are linear. For symmetric players in a network congestion game, or if player strategies correspond to bases of a matroid, the equilibrium can be computed in polynomial time~\cite{AckermannRV06,fabrikant}.

\section{Preliminaries}
An \emph{atomic splittable routing game} (ASRG) $\varGamma = (G=(V,E),$ $(v_i,s_i,t_i)_{i \in [n]},$ $(l_e)_{e \in E}))$ is defined on a directed network $G = (V,E)$ with $n$ players. Each player $i$ wants to send $v_i$ units of flow from $s_i$ to $t_i$, where $v_i$ is the \emph{demand} of player $i$. Each edge $e$ has a cost function $l_e(x)$ which is non-negative, increasing, convex and differentiable. Players are indexed so that $v_1\geq v_2\geq \ldots \geq v_n$, and the total demand $V := \sum_i v_i$.  Vector $f^i$ denotes the flow of player $i$. By abuse of notation, we say vector $f$ is the flow on the network with $n$ players such that $f_e^i$ denotes the amount of flow player $i$ sends along the edge $e$, and $f_e:=\sum_i f^i_e$ is the total flow on edge $e$. 

Given a flow $f=(f^1,f^2,\ldots,f^n)$ for $n$ players, player $i$ incurs a cost  $\mathcal{C}^i_{e}(f):= f_{e}^{i}l_{e}(f_{e})$ on the edge $e$. His total cost is $\mathcal{C}^i(f)=\sum_{e\in E}\mathcal{C}^i_{e}(f)$. Each player's objective is to minimize his cost, given the flow of the other players. We say a flow $f$ is at \emph{equilibrium}\footnote{More specifically, a pure Nash equilibrium.} if no player can unilaterally change his flow and reduce his total cost. More formally,

\begin{definition}
In an ASRG a flow $f=(f^1,f^2,\ldots,f^n)$ is a Nash Equilibrium flow if for every player $i$ and every flow $g=(f^1,f^2,\ldots,$ $f^{i-1},g^i,f^{i+1},\ldots,f^n)$, where $g^i$ is a flow of value $v_i$, $\mathcal{C}^i(f)\leq \mathcal{C}^i(g).$
\end{definition}

The equilibrium flow can be characterized in terms of the marginal costs of each player. Intuitively, the marginal cost for a player on a path is the increase in cost for the player when he increases his flow on the path by a small amount.

\begin{definition}
 Given a flow $f$, the marginal cost for the player $i$ on path $p$ is given by
 \[L_{p}^{i}(f)=\sum_{e\in p}l_{e}(f_{e})+f_{e}^{i}l_{e}'(f_{e})\]
\end{definition}

By applying the \emph{Karush-Kuhn-Tucker} conditions~\cite{kuhn} for player $i$'s minimization problem, we get the following lemma which characterizes the equilibrium using marginal costs.

\begin{lemma}\label{umangle1}
 Flow $f=(f^1,f^2,\ldots,f^n)$ is a Nash equilibrium flow iff for any player $i$ and any two directed paths $p$ and $q$ between $s$ and $t$ such that $f_e^i>0\ \forall e\in p$, $L_{p}^{i}(f)\leq L_{q}^{i}(f)$.
\end{lemma}

Lemma~\ref{umangle1} says that for any player $i$ at equilibrium, the marginal delay $L_p^i(f)$ on all paths $p$ such that $f_e^i>0\ \forall e\in p$ is equal, and is the minimum over all $s$-$t$ paths. In a network of parallel edges, every edge is an $s$-$t$ path, hence the condition holds at equilibrium with edges replacing paths.

We will frequently use the \emph{support} of a player, where given a flow $f=(f^1,f^2,\ldots,f^n)$, the support of player $i$, $S_i$ is defined as the set edges with $f_e^i>0$.

Swamy studies the use of edge tolls to enforce a particular flow as equilibrium~\cite{Swamy12}. However, if we start with an equilibrium flow, the tolls required are identically zero. The following theorem regarding equilibrium computation is then implicit, and will be useful to us in Section~\ref{sec:typesexp}.

\begin{theorem}[\cite{Swamy12}]
For ASRG $\varGamma$, let $(h_e)_{e \in E}$ be the total flow on each edge at an equilibrium. Then given the total flow $h$, the equilibrium flow for each player can be obtained in polynomial time by solving a convex quadratic program.
\label{thm:swamy}
\end{theorem}

We make the following smoothness assumptions on edge cost functions.
\begin{enumerate}
\item Cost functions are continuously differentiable, nonnegative, convex, and increasing. 
\item There is a constant $\Psi \ge V$ that satisfies:
\[
\Psi \ge \max_{e \in E, \, x \in [0,V]} \left \{ l_e(x), \, l_e'(x), \, l_e''(x), \, \frac{1}{l_e'(x)} \right \}
\]
\end{enumerate}

\noindent By the first assumption, the edge marginal cost $L_e^i(f)$ is strictly increasing, both with the total flow $f_e$ and player $i$'s flow $f_e^i$. Define $L^i(f):=\min_{e\in E}L^i_e(f).$ Hence if $\vec{f}$, $\vec{f}'$ are two equilibrium flows, and for some player $i$ and edge $e$ $f_e \ge f_e'$, $f_e^i \ge {f_e'}^i$, and $f_e^i > 0$, then
\[
L^i(f) = L_e^i(f) \ge L_e^i(f') \ge L^i(f')
\]

\noindent and the second inequality is strict if $f_e > f_e'$ or $f_e^i > {f_e^i}'$. We frequently use this inequality in our proofs.

Also observe that for each edge $e$ and flow values $x$, $y \in [0,V]$, the following properties of the edge cost functions hold.
\begin{align*}
\lvert x-y \vert ~ \le  ~ \delta  \Rightarrow ~ \lvert l_e(x) - l_e(y) \rvert \le \delta \Psi \quad \mbox{ and } \quad \lvert l_e(x) - l_e(y) \rvert ~\le~ \delta ~ \Rightarrow ~ \lvert x - y \rvert ~\le~ \delta \Psi
\end{align*}

\section{An Algorithm with Complexity Exponential in the Number of Players}
\label{sec:playersexp}


To convey the main ideas of the algorithm, we will ignore issues regarding finite precision computation in this section. In particular, we assume the algorithm carries out binary search to infinite precision, and show that such an algorithm computes the exact equilibrium. In Section~\ref{sec:implem}, we then give an implementation of the algorithm. We show that our implementation computes an $\epsilon$-equilibrium in time $O\left(mn^2 \left(\log (n\Psi/\epsilon) \right)^n\right)$ (Theorem~\ref{thm:implement}).

Recall that the equilibrium in case of parallel edges is unique~\cite{shimkin}. We start with an outline of the algorithm. Our first idea is to reduce the problem of equilibrium computation, to finding the \emph{marginal costs} at equilibrium. We give a function $\graphflow$ that at a high level, given a vector of marginal costs $\vec{M} = (M^1, \dots M^n)$, returns a vector of demands $\vec{w} = (w_1, \dots, w_n)$ and a flow vector $\vec{f} = (f_e^i)_{e \in E, i \in [n]}$ so that (1) $\vec{f}$ is the equilibrium flow for the demand vector $\vec{w}$, and (2) for each player $i$, the marginal cost $L^i(f) = M^i$. That is, the marginal costs for the players at equilibrium are given by the input vector $\vec{M}$. We show that in fact each marginal cost vector $\vec{M}$ maps to a unique (demand, flow) pair that satisfies these conditions. Hence given marginal costs at equilibrium, the function must return the correct demands $(v_i)_{i \in [n]}$, and the required equilibrium flow. Thus, our problem reduces to finding a marginal cost vector $\vec{M}$ for which $\graphflow$ returns the correct demand vector $\vec{v} = (v^1, v^2, \dots, v^n)$. We say a demand $w^i$ for player $i$ is \emph{correct} if $w^i = v^i$.

Since only the marginal costs and demands matter to us, we can think of $\graphflow$ as a function from marginal cost vectors to demand vectors. We then give a high-dimensional binary search algorithm that computes the required marginal cost vector. This proceeds in a number of steps. We first show that the function $\graphflow$ is continuous, and is monotone in a strict sense: if we increase the marginal cost of a player, then the demand for this player increases, and the demand for every other player decreases. This allows us to show in Lemma~\ref{lem:existencenplayers} that given any marginal costs for the first $n-k$ players, there exist marginal costs for the remaining $k$ players so that the demands returned by $\graphflow$ for these remaining players is correct. This lemma allows us to ignore first $n-k$ players, and focus on the last $k$ players, since no matter what marginal costs we choose for the first $n-k$ players, we can find marginal costs for the last $k$ players that give the correct demand for these players.

The crux of our binary search algorithm is then Lemma~\ref{lem:monotonenplayers}, which says the following. Suppose we are given two marginal cost vectors $\vec{M}$ and $\vec{M'}$ that differ only in their last $k$ coordinates, and for which the demands of the last $k-1$ players is equal. Thus, $M^i = {M^i}'$ for all players $i < k$, and the demands returned by $\graphflow(\vec{M})$, $\graphflow(\vec{M}')$ are equal for all players $i > k$. Suppose for the $k$th player, the demand with marginal costs $\vec{M}$ is higher than the demand with marginal costs $\vec{M}'$. Then the lemma says that $k$'s marginal cost in $\vec{M}$ must be higher than in $\vec{M'}$, i.e., $M^k > {M^k}'$. This lemma allows us to give a recursive binary search procedure. For a player $k$, the procedure fixes a marginal cost $M^k$, and finds marginal costs for players $i > k$ so that these players have the correct demand. By Lemma~\ref{lem:existencenplayers}, we know that such marginal costs exist. With these marginal costs, if the demand for player $k$ is greater than $v_k$, then by Lemma~\ref{lem:monotonenplayers} $M^k$ is too large. We then reduce $M^k$, and continue.

\begin{algorithm}[!ht]
\caption{GraphFlow($\vec{M}$)}\label{algo:graphflow}
\begin{algorithmic}[1]
\Require{Vector $\vec{M}=(M^i)_{i \in [n]}$ of nonnegative real values}
\Ensure{Flow $\vec{f} = (f_i(e))_{i \in [n], e \in E}$ and demands $\vec{w} = (w_i)_{i \in [n]}$ so that $w_i = \lvert f^i \rvert$ and $\vec{f}$ is an equilibrium flow for demands $\vec{w}$ with marginal costs $\vec{M}$.}
\State Assume that $M^1 \ge M^2 \ge \dots \ge M^n$, else renumber the vector components so that this holds.
\For {each edge $e \in E$} \label{line:graphflow1}
	\State {$f_e^i = 0$ for each player $i \in [n]$}
	\If {$l_e(0) \ge M^1$}
		\State {$S_e \gets \emptyset$; continue with the next edge} \label{line:graphflowse1}
	\EndIf
	\For {$k = 1 \to n$}
		\State $S = [k]$
		\State {Let $x_e$ be the unique solution to $k l_e(x) + x l_e'(x) = \sum_{i \in S} M^i$} \label{line:graphflowxe} \Comment{Since $l_e(x)$ is strictly increasing and convex, the solution is unique}
		\State $f_e^i = \frac{M^i - l_e(x_e)}{l_e'(x_e)}$ for each player $i \in S$ \label{line:graphflowfei} \Comment{Note that $\sum_{i \in S} f_e^i = x_e$}
		\If {$(f_e^i \ge 0$ for all $i \in S) \textbf{ and } (k=n \textbf{ or } M^{k+1} \le l_e(x_e)$)}
			\State $f_e \gets x_e$, $S_e \gets S$, continue with the next edge \label{line:graphflowse}
		\EndIf
	\EndFor
\EndFor
\State {$w_i \gets \sum_e f_e^i$ for each player $i$; \Return {($\vec{f}, \vec{w}$)}}
\end{algorithmic}
\end{algorithm}

Algorithm~\ref{algo:graphflow} describes the function $\graphflow$. The algorithm considers each edge in turn. For an edge $e$, it tries to find a subset of players $S \subseteq [n]$ and flows $f_e^i$ so that, for all players $i \in S$, $L_e^i(f) = M^i$, and for all players not in $S$, $f_e^i = 0$ and $M^i \le L_e^i(f)$. The set $S$ can be obtained in $O(n)$ time by adding players to $S$ in decreasing order of marginal costs $M^i$. Given a set $S$, summing the equalities $L_e^i(f) = M^i$, we get the following equation with variable $f_e$:
\begin{align*}
|S| \, l_e(f_e) + f_e l_e'(f_e) ~=~ \sum_{i \in S} M^i \, .
\end{align*}
\noindent Noting that the left-hand side is strictly increasing in $f_e$, we can solve this equation for $f_e$ using binary search. This gives us the total flow on the edge $f_e$. We can then obtain the flow for each player by solving, for each player $i \in S$, the following equation:
\begin{align*}
f_e^i ~= ~ \frac{M^i - l_e(f_e)}{l_e'(f_e)} \, .
\end{align*}
\noindent We set $f_e^i = 0$ for all players not in $S$. It can be checked that $\sum_{i \in S} f_e^i = f_e$. If $f_e^i \ge 0$ for all players, and $L_e^i(f) = l_e(f_e) \ge M^i$ for all players not in $S$, we move on to the next edge. Else, we add the next player with lower marginal cost $M^i$ to the set $S$, and recompute $f_e$.

We first establish in Claims~\ref{claim:graphfloweq},~\ref{claim:unique}, and~\ref{claim:continuous} that the algorithm is correct, and gives a continuous map from marginal cost vectors to demand vectors.

\begin{claim}
Given $\vec{M}$, assume w.l.o.g. that $M^1 \ge M^2 \ge \dots \ge M^n$. 
Then GraphFlow($\vec{M}$) returns flow $\vec{f}$ and demands $\vec{w}$ so that, on each edge $e$ and for each player $i$,
\begin{enumerate}
\item if $f_e^i = 0$ then $L_e^i(f) \ge M^i$
\item if $f_e^i > 0$ then $L_e^i(f) = M^i$
\end{enumerate}
\noindent Thus, $\vec{f}$ is an equilibrium flow for values $\vec{w}$, and if $w_i > 0$ then $M^i = L^i(f)$.
\label{claim:graphfloweq}
\end{claim}
\begin{proof}
Fix an edge $e$. We first show that for this edge, $S_e$ is defined, i.e., either Line~\ref{line:graphflowse1} or Line~\ref{line:graphflowse} is executed. If Line~\ref{line:graphflowse1} is not executed, then $M^1 > l_e(0)$. More generally, for $k \in \{1, \dots, n\}$, let $x_e(k)$ be the value obtained for $x$ in Line~\ref{line:graphflowxe}. For each player $i \le k$, define
\[
f_e^i(k) = (M^i - l_e(x_e(k))/l_e'(x_e(k)) \,
\]
\noindent and $f_e^i(k) = 0$ for $i > k$. Recall that by assumption, $M^1 \ge M^2 \ge \dots \ge M^n$. We will show that for some value of $k$, Line~\ref{line:graphflowse} gets executed. That is, each $f_e^i(k) \ge 0$, and either $k=n$ or $M^{k+1} \le l_e(x_e(k))$.

Let $k = 1$. Since $M^1 > l_e(0)$, $x_e(1) > 0$, and hence $f_e^1(1)$ $= x_e(1)$ $> 0$. More generally, suppose that for some $k \le n$, $f_e^i(k) \ge 0$ for all players $i \le k$. Then either Line~\ref{line:graphflowse} is executed for this value of $k$, or $M^{k+1} > l_e(x_e(k))$. We will show that in the latter case, $f_e^i(k+1) \ge 0$ for $i \le k+1$. To see this, consider the expressions for $x_e(k)$ and $x_e(k+1)$:
\begin{align*}
M^1 + \dots + M^k & = k \, l_e(x_e(k)) + x_e(k) \,l_e'(x_e(k)) \\
M^1 + \dots + M^k + M^{k+1} & = (k+1) \, l_e(x_e(k+1)) + x_e(k+1) \,l_e'(x_e(k+1)) \\
	&  \ge l_e(x_e(k+1)) + k \, l_e(x_e(k)) + x_e(k) \,l_e'(x_e(k))
\end{align*}
\noindent where the inequality is because $x_e(k+1) \ge x_e(k)$. Subtracting the first expression from the second then gives us that $M^{k+1} \ge l_e(x_e(k+1))$, and hence $f_e^{k+1}(k+1) \ge 0$. Since $M^{k+1} \le M^i$ for all $i \le k+1$, $f_e^i(k+1) \ge 0$ for all $i \le k+1$.

Thus, when $k = n$, either Line~\ref{line:graphflowse} has previously been executed, or $f_e^i(n) \ge 0$ for all player $i$. In the latter case, the `if' condition holds true, and Line~\ref{line:graphflowse} must be executed in this iteration.

Now fix a player $i$. We consider the following cases.

\noindent \textbf{Case 1:} $S_e$ is empty. This is true iff $M^1 \le l_e(0)$. In this case, $f_e^i = 0$, and $M^i \le M^1 \le l_e(0) = L_e^i(f)$.

\noindent\textbf{Case 2a:} $S_e$ is not empty, and $i$ is in $S_e$. Then $f_e^i \ge 0$, and since $f_e = \sum_{i \in [n]}f_e^i$ and by the expression for $f_e^i$ in Line~\ref{line:graphflowfei}, we obtain
\[
L_e^i(f) ~ = ~ l_e(f_e) + f_e^i l_e(f_e) ~=~ M^i \, .
\]

\noindent \textbf{Case 2b:} $S_e$ is not empty, and $i$ is not in $S_e$. Then $f_e^i = 0$, $M^i \le l_e(f_e)$, and $l_e(f_e) = L_e^i(f)$, completing the proof.
\end{proof}

\begin{corollary}
\label{cor:graphfloweq}
Given $\vec{M}$, GraphFlow($\vec{M}$) returns flow $\vec{f}$ and demands $\vec{w}$ so that $\vec{f}$ is an equilibrium flow for values $\vec{w}$, $L^i(f) \le M^i$, and if $w_i > 0$ then $L^i(f) = M^i$.
\end{corollary}

\begin{proof}
Fix a player $i$. If $w_i = 0$, then $f_e^i = 0$ on each edge and this is trivially an equilibrium flow, and $M^i \le L_e^i(f)$ on every edge by the claim. If $w_i > 0$ then on each edge $e$ with $f_e^i > 0$, $L_e^i(f) = M^i$ and $M^i \le L_{e'}^i(f)$ for every edge $e'$ by the Claim. Hence $L_e^i(f) = M^i = \min_{e'} L_{e'}^i(f)$. Hence player $i$ only puts positive flow on minimum marginal cost edges, and $\vec{f}$ is an equilibrium.
\end{proof}

\begin{claim}
For each vector $\vec{M}$ of marginal costs, there is a unique pair of vectors $(\vec{w}, \vec{f})$ so that:
\begin{enumerate}
\item $\vec{f}$ is the equilibrium flow for demands $\vec{w}$, and
\item for each player $i$, $L^i(f) \ge M^i$. If $w_i > 0$, then $L^i(f) = M^i$.
\end{enumerate}
\label{claim:unique}
\end{claim}

\begin{proof} 
Assume for a contradiction that $(\vec{w}, \vec{f})$ and $(\vec{w}', \vec{f}')$ satisfy the properties, and the two are unequal. Since the equilibrium in parallel edges is unique, and $\vec{f}$, $\vec{f}'$ must be equilibrium flows for $\vec{w}$, $\vec{w}'$ respectively, if $\vec{w} = \vec{w}'$ then $\vec{f} = \vec{f}'$. Hence, $\vec{w}$, $\vec{w}'$ are unequal. Thus for some player $i$, $w_i \neq w_i'$. Assume that $w_i > w_i'$. Then since $w_i > 0$, $L^i(f') \ge M^i = L^i(f)$. Also for some edge $e$, $f_e^i > {f_e^i}'$. If $f_e \ge f_e'$, then $L_e^i(f) > L_e^i(f')$, and hence
\[
M^i ~=~ L^i(f) ~=~ L_e^i(f) ~>~ L_e^i(f') ~\ge~ L^i(f') ~\ge~ M^i  \, ,
\]
\noindent and thus $M^i > M^i$, which is a contradiction. If $f_e < f_e'$, then there is a player $j$ with $f_e^j < {f_e^j}'$ and $w_j' > 0$. By the same argument as earlier, we again get a contradiction. Thus, for each vector $\vec{M}$ of marginal costs, there is a unique pair of vectors $(\vec{w}, \vec{f})$ that satisfy the conditions in the claim.
\end{proof}

\begin{corollary}
For a demand vector $\vec{w}$, let $\vec{f}$ be the equilibrium flow, and $M^i = L^i(f)$ be the marginal costs of the players at equilibrium. Then the function $\graphflow(M^1, \dots, M^n)$ returns $(\vec{w}, \vec{f})$ as the output.
\label{cor:unique}
\end{corollary}
\begin{proof}
The corollary follows by observing that $(\vec{w}, \vec{f})$ satisfy the conditions in Claim~\ref{claim:unique}, and by the claim is the only pair of vectors that satisfies these conditions; and that by Corollary~\ref{cor:graphfloweq}, $\graphflow(M^1, \dots, M^n)$ returns a pair of vectors that must satisfy the conditions in Claim~\ref{claim:unique}.
\end{proof}


\begin{claim}
Given marginal costs $\vec{M}$, $\vec{M}'$ so that for player 1, $\lvert M^1 - {M^1}' \rvert \le \epsilon$, and $M^j = {M^j}'$ for all players $j > 1$, let $(\vec{f}, \vec{w})$ and $(\vec{f}', \vec{w}')$ be the flows and demands returned by GraphFlow. Then for each player $i$, $\lvert f_e^i - {f_e^i}' \rvert \le \epsilon'$, where $\epsilon' = 2n \Psi \epsilon$. Hence, for each player $i$, $\lvert w_i - w_i' \rvert \le m \epsilon'$.
\label{claim:continuous}
\end{claim}
\begin{proof}
Assume for a contradiction that for some player $k$ and an edge $e$, $f_e^k \ge {f_e^k}' + \epsilon'$. By Claim~\ref{claim:graphfloweq}, flows $\vec{f}$ and $\vec{f}'$ are the equilibrium flows for demands $\vec{w}$ and $\vec{w}'$ respectively. Further, if for some edge $e$ and player $i$ the flow $f_e^i > 0$, then $M^i = L^i(f) = L_e^i(f)$. Similarly, if ${f_e^i}' > 0$, then ${M^i}' = L^i(f') = L_e^i(f')$.

We will consider two cases. In the first case, assume $f_e \ge f_e'$. Then since $f_e^k > 0$,
\[
M^k ~=~ L_e^k(f) ~=~ l_e(f_e) + f_e^k l_e'(f_e) ~\ge~ l_e(f_e') + ({f_e^k}' + \epsilon') l_e(f_e') ~\ge~ L_e^k(f') + \epsilon'/\Psi \ge {M^k}' + \epsilon'/ \Psi \, .
\]
\noindent where the second inequality is because $l_e(f_e') \ge 1/\Psi$, and the last inequality is by Claim~\ref{claim:graphfloweq}. This is a contradiction, since for each player $i$, $\lvert M^i - {M^i}' \rvert \le \epsilon =  \epsilon' / (2n \Psi)$.

Now consider the case that $f_e < f_e'$. Since $f_e^k \ge {f_e^k}' + \epsilon'$, there exists a player $j$ so that $f_e^j \le {f_e^j}' - (\epsilon'/n)$. Then proceeding similarly as above, since now ${f_e^j}' > 0$, we get that
\[
{M^j}' ~\ge~ M^j + \frac{\epsilon'}{n \Psi} \, .
\]
\noindent This is again a contradiction, since $\lvert M^i - {M^i}' \rvert \le \epsilon' / (2n \Psi)$ for each player $i$.
\end{proof}
We note that since in the claim the choice of player 1 is arbitrary, this holds for any player. The claim then shows that the function GraphFlow is continuous.


In the remainder of the discussion, given a vector of marginal costs $\vec{M}$, we will primarily be concerned with the demands $\vec{w}$ returned by the function GraphFlow. We therefore define the functions \graphval{i} for each player $i \in [n]$. Function \graphval{i} takes as input a vector $\vec{M}$ of marginal costs for the players, and returns the demand $w_i$, the $i$th component of the demand vector $\vec{w}$ returned by GraphFlow($\vec{M}$). Claim~\ref{claim:monotone} now shows that the function $\graphflow$ is monotone: if we increase the input marginal cost of a player, that player's demand goes up, while the demand for all the other players goes down. This is crucial in establishing existence of marginal costs for a subset of players (Lemma~\ref{lem:existencenplayers}), and in our binary search algorithm later on.

\begin{claim}
Consider marginal cost vectors $\vec{M}$ and $\vec{M}'$ that differ only in their first coordinate, so that $\vec{M} = (M^i)_{i \in [n]}$ and $\vec{M}' = ({M^1}', M^2, \dots, M^n)$. For each player $i$, let $w_i = \graphval{i}(\vec{M})$, and $w_i' = \graphval{i}(\vec{M}')$. If ${M^1}' > M^1$, then the following hold true as well:
\begin{enumerate}
\item $w_1' \ge w_1$,
\item if $w_1' > 0$, then $w_1' > w_1$,
\item $w_i' \le w_i$ for $i > 1$, and
\item for any subset of players $P$ containing player 1, $\sum_{i \in P} w_i' \ge \sum_{i \in P} w_i$.
\end{enumerate}
\label{claim:monotone}
\end{claim}
\begin{proof}
We first claim that on each edge, the total flow $f_e' \ge f_e$. This will be used to prove each of the statements in the claim. Assume for a contradiction that there is an edge $e$ so that $f_e' < f_e$, and hence a player $i$ with ${f_e^i}' < f_e^i$. Since $f_e^i > 0$, by Claim~\ref{claim:graphfloweq}, $M^i = L^i(f)$, which is also equal to $L_e^i(f)$. Hence
\[
M^i ~=~ L^i(f) ~=~ L_e^i(f) ~>~ L_e^i(f') \ge {M^i}' \, .
\]
\noindent But for each player the marginal cost ${M^i}'$ is at least as large as $M^i$, hence this is a contradiction.

To prove the third statement in the claim, we show that in fact for each each edge $e$ and each player $i > 1$, the flow ${f_e^i}' \le f_e^i$. If not, then ${f_e^i}' > f_e^i$ and $f_e' \ge f_e$, and a similar calculation as previously shows that ${M^i}' > M^i$,  which is again a contradiction. Thus for every player except the first, the flow on every edge is nonincreasing, and hence the demand $w_i' \le w_i$ for $i > 1$. Since the total flow on every edge is nondecreasing, player $1$'s flow is nondecreasing on every edge, and hence $w_1' \ge w_1$, which proves the first statement.

For the second statement in the claim, we have already shown that for each edge, $f_e' \ge f_e$, ${f_e^1}' \ge f_e^1$, and ${f_e^i}' \le f_e^i$ for each player $i > 1$. Since $w_1' > 0$, there is an edge $e$ with ${f_e^1}' > 0$. If ${f_e^1}' > f_e^1$ on this edge, then clearly $w_1' > w_1$ as required. Suppose for a contradiction that ${f_e^1}' = f_e^1 >0$. Then since for all other players, ${f_e^i}' \le f_e^i$ but $f_e' \ge f_e$, it must be true that $f_e' = f_e$. Then by Claim~\ref{claim:graphfloweq}, ${M^1}' = L_e^i(f') = L_e^i(f)  = M^1$, which is a contradiction.

The fourth statement is obtained from two observations. Firstly, the total flow $\sum_i w_i' \ge \sum_i w_i$, since otherwise on some edge $f_e' < f_e$. Second, by the third statement in the claim, for any set of players $P'$ that excludes player 1, $\sum_{ i \in P'} w_i' \le \sum_{i \in P'} w_i$. Thus, subtracting the two inequalities, we get $\sum_{i \not \in P'} w_i' \ge \sum_{i \not \in P'} w_i$, as required.
\end{proof}


Consider the game with cost functions as in $\varGamma$, but with $n$ players, each with demand $V$. Since this is a symmetric game, the equilibrium flow can be computed in polynomial time~\cite{cominetti}. We define $\Lambda$ to be the marginal cost of each player at this equilibrium.

\begin{lemma}
Let $S \subseteq [n]$ be a subset of the players. Given strictly positive input demands $\hat{w}_i \le V$ for players $i \in S$ and marginal costs $M^i \le \Lambda$ for players $i \not \in S$, there exist marginal costs $\hat{M}^i \le \Lambda$ for the players in $S$ so that, given input $((\hat{M}^i)_{i \in S},(M^i)_{i \not \in S})$, $\graphval{i}$ returns $\hat{w}_i$ as the demand for players $i \in S$.
\label{lem:existencenplayers}
\end{lemma}

We first prove two other claims, which establish bounds of $[0, \Lambda]$ on the marginal cost of each player, and the second which proves Lemma~\ref{lem:existencenplayers} for the special case when $|S| = 1$.

\begin{claim}
Consider the marginal cost vector $\vec{M}$ where the first player has marginal cost $M^1 = \Lambda$, and all other players have marginal costs $M^i$ in the interval $[0, \Lambda]$. Then $\graphval{1}(\vec{M}) \ge V$.
\label{claim:bounds}
\end{claim}

\begin{proof}
By definition, if each player had marginal cost $\Lambda$, then the demand for each player would be $V$. Since we do not change the marginal cost of the first player and decrease the marginal costs for the other players, by Claim~\ref{claim:monotone} the demand for player $1$ can only increase. This completes the proof.
\end{proof}


\begin{claim}
Given marginal costs $M^2$, $\dots$, $M^n$ for all players except the first, each in 	the interval $[0, \Lambda]$ and given a desired demand $\hat{w} \le V$ for the first player, there exists $\hat{M} \in [0, \Lambda]$ so that for the marginal cost vector $\vec{M} = (\hat{M}, M^2, \dots M^n)$, function $\graphval{1}$ returns demand $\hat{w}$ for the first player.
\label{claim:existence1player}
\end{claim}

\begin{proof}
Consider the marginal cost vectors $(0, M^2, \dots, M^n)$ and $(\Lambda, M^2, \dots, M^n)$. In the first case, $\graphval{1}$ returns $0$, while in the second case, it returns a value at least $V$ by Claim~\ref{claim:bounds}. Further, by Claim~\ref{claim:continuous}, as we vary the marginal cost of player 1, the demands computed by function $\graphval{1}$ vary continuously. Hence, there must exist some value $\hat{M} \in [0, \Lambda]$ for which $\graphval{1}$ returns $\hat{w}$ as the demand for player 1.
\end{proof}

\begin{proof}[Proof of Claim~\ref{lem:existencenplayers}]
The proof is by constructing a sequence $\vec{M}(0), \vec{M}(1), \dots$ of marginal cost vectors with each component nondecreasing and bounded from above by $\Lambda$. Then by the monotone convergence theorem, the sequence has a limit $\vec{N}$. We will show that for all players $i \in S$, $\graphval{i}(\vec{N}) = \hat{w}_i$.

The sequence is constructed as follows. We index the steps in the sequence by $t$. 
For each player $i \not \in S$, we set $M^i(t) = M^i$ for each step $t$. Hence the components corresponding to players not in $S$ do not change and are bounded by $\Lambda$. 

Initially, $M^i(0) = 0$ for players $i \in S$. Given $\vec{M}(t)$ with each component at most $\Lambda$, we obtain $\vec{M}(t+1)$ as follows. For each player $i$, we define $M^i(t+1)$ as the value $M$ so that, given input $(M, (M^j(t))_{j \neq i})$, the function $\graphval{1}$ returns $\hat{w}_i$ as the demand for the first player. By Claim~\ref{claim:existence1player}, such a value $M$ exists. Further, $M^i(t+1) \le \Lambda$ by the claim. Hence each component of $\vec{M}(t+1)$ is at most $\Lambda$.

We now show that each component in the sequence is nondecreasing. For players not in $S$, by construction, the corresponding components in the marginal cost vectors $\vec{M}(t)$ do not change. Hence we need only concern ourselves with players in $S$. Further, this is clearly true for the first step: $M^i(1) \ge M^i(0) = 0$. Suppose this is true for step $t$, i.e., $M(t) \ge M(t-1)$. Fix a player $i \in S$, and assume for a contradiction that $M^i(t+1) < M^i(t)$. By construction, given input $(M^i(t), (M^j(t-1))_{j \neq i})$, \graphval{1} outputs $\hat{w}_i > 0$. By Claim~\ref{claim:monotone}, since $M(t) \ge M(t-1)$, given input $(M^i(t), (M^j(t))_{j \neq i})$, \graphval{1} returns a value at most $\hat{w}_i$. Finally, since $\hat{w}_i > 0$ and $M^i(t+1) < M^i(t)$, given input $(M^i(t+1), (M^j(t))_{j \neq i})$, \graphval{1} returns a value strictly less than $\hat{w}_i$, which is a contradiction, since by construction that value returned should be $\hat{w}_i$. Hence, the components of the marginal cost vectors $\vec{M}(t)$ are monotone, bounded by $\Lambda$, and hence the sequence has a limit $\vec{N}$. Since the component for each player $i$ not in $S$ is equal to $M^i$ throughout the sequence, this is true of $\vec{N}$ as well, and $N^i = M^i$ for each player $i \not \in S$.

We now show that for each player $i \in S$, \graphval{i}$(\vec{N}) = \hat{w}_i$. Assume for a contradiction that \graphval{i}$(\vec{N}) = w_i' \neq \hat{w}_i$, and that $|w_i' - \hat{w}_i| = \delta$. Let step $T$ be such that, for all players $j$ and steps $t \ge T$, $|M^j(t) - N^j| \le \delta/(6mn^2\Psi)$. Then, by Claim~\ref{claim:continuous},
\[
\lvert \graphval{i}(\vec{M}(T+1)) - w_i'| \le \delta/3 \, ,
\]

\noindent and hence,
\begin{align}
\lvert \graphval{i}(\vec{M}(T+1)) - \hat{w}_i| \ge 2\delta/3 \, . \label{eqn:existencenplayer}
\end{align}

\noindent By construction, $\graphval{1}(M^i(T+1), (M^j(T))_{j \neq i}) = \hat{w}_i$. Since $\Vert M(T+1) - M(T) \Vert_1 \le \delta/(6mn\Psi)$, it then follows from Claim~\ref{claim:continuous} that
\[
\lvert \graphval{1}(\vec{M}(T+1)) -  \hat{w}_i\rvert \le \delta/3 \, ,
\]
\noindent which contradicts~\eqref{eqn:existencenplayer}. Hence, for each player $i \in S$, \graphval{i}$(\vec{N}) = \hat{w}_i$.
\end{proof}

\begin{lemma}
Given a player $k$, and two marginal cost vectors $\vec{M}$ and $\vec{M}'$ that satisfy the following properties:
\begin{enumerate}
\item for all players $i < k$, $M^i = {M^i}'$,
\item for all players $i > k$, $\graphval{i}(\vec{M}) = \graphval{i}(\vec{M}')$.
\end{enumerate}
Let $w_k = \graphval{k}(\vec{M})$, and $w_k' = \graphval{k}(\vec{M}')$. If $w_k > w_k'$, then $M^k > {M^k}'$.
\label{lem:monotonenplayers}
\end{lemma}

\begin{proof}
Let $M^k \le {M^k}'$, and let $P$ be the set of players $\{i \ge k: M^i \le {M^i}'\}$. Thus $k \in P$. We will show that $w_k \le w_k'$. The proof proceeds by changing the marginal cost of each player in order from $M^i$ to ${M^i}'$, and considering the effect on total demand of players in $P$. We show that in this process, the total demand of these players does not increase, and hence
\begin{align}
\sum_{i \in P} \graphval{i}(\vec{M}) ~ \le ~ \sum_{i \in P} \graphval{i}(\vec{M}') \, . \label{eqn:Pincreases}
\end{align}
\noindent The expression on the left equals $w_k + \sum_{i \in P, i \neq k} w_i$, while the expression on the right equals $w_k' + \sum_{i \in P, i \neq k} w_i$ since for players $i > k$, $\graphval{i}(\vec{M}) = \graphval{i}(\vec{M}')$. Hence, this will show that $w_k \le w_k'$, as required.

\noindent We now need to prove~\eqref{eqn:Pincreases}. Our proof uses Claim~\ref{claim:monotone}. Formally, let $\vec{M}(t) = (({M^i}')_{i \le t},(M^i)_{i > t})$. Then $\vec{M}(0) = \vec{M}$, and $\vec{M}(n) = \vec{M'}$. We will show that
\begin{align}
\sum_{i \in P} \graphval{i}(\vec{M}(t)) ~ \le ~ \sum_{i \in P} \graphval{i}(\vec{M}(t+1)) \, . \label{eqn:Pincreases2}
\end{align}
\noindent For each player $t \in [n]$, there are three cases: either (1) $M^t = {M^t}'$, (2) $M^t \le {M^t}'$ and $t \in P$, or (3) $M^t > {M^t}'$ and $t \not \in P$. We consider these three cases separately. In the first case, $\vec{M}(t) = \vec{M}(t+1)$, and~\eqref{eqn:Pincreases2} clearly holds. In the second case, by Claim~\ref{claim:monotone} and since $t \in P$,~\eqref{eqn:Pincreases2} holds. In the third case, again by Claim~\ref{claim:monotone}, for all $i \neq t$ the demand either increases or remains the same. Since $t \not \in P$, the total demand of players in $P$ either increases or remains the same, and hence~\eqref{eqn:Pincreases2} holds. This completes the proof.
\end{proof}

We now give our algorithm for obtaining the ``correct'' marginal costs $\vec{M}$, so that $\graphflow(\vec{M})$ returns the required equilibrium flow. The algorithm is recursive. 

For each player $k$, it picks a candidate marginal cost $M^k$, and then recursively calls itself to find marginal costs $M^i$ for players $i > k$ so that the demands for these players $i > k$ is correct. If the demand for player $k$ itself is too large, it reduces $M^k$, and otherwise increases $M^k$. Thus the algorithm conducts a binary search to find the correct marginal cost for player $k$, and in each iteration calls itself to determine correct marginal values for players $i > k$.


\begin{algorithm}[!ht]
\caption{\eqmcost{k}($(M^1, \dots, M^{k-1}$))}\label{algo:eqmcost}
\begin{algorithmic}[1]
\Require{Vector $(M^1, \dots, M^{k-1})$, with each component $M^i \in [0,\Lambda]$} \Comment{If $i=1$, there is no input required.}
\Ensure{Vector $(\vec{M})$ of marginal costs so that the first $k-1$ marginal costs are equal to the inputs, and for players $i \ge k$, the demand $\graphval{i}(\vec{M}) = v_i$.}
\If{$k = n$}
	\State {Using binary search in $[0,\Lambda]$, find $M$ so that $\graphval{n}((M^i)_{i < n}, M) = v_n$. \Return $M$.}
\EndIf
\State {$\low \gets 0$, $\high \gets \Lambda$, $\med \gets (\low + \high)/2$}
\State {$(M^{k+1}, \dots, M^n) \gets \eqmcost{k+1}(M^1, \dots, M^{k-1}, \med)$} \label{line:finalloop} \Comment{Call \eqmcost{k+1} to get marginal costs for the remaining player $k+1, \dots, n$ so that the demand for these players is correct}
\If {$(\graphval{k}((M^i)_{i < k}, \med, (M^i)_{i > k}) = v_k)$}
		\State \Return {$(\med, (M^i)_{i > k})$}
\ElsIf {$(\graphval{k}((M^i)_{i < k}, \med, (M^i)_{i > k}) > v_k)$}
		\State{$\high \gets \med$, $\med \gets (\low + \high)/2$, goto~\ref{line:finalloop}}
\ElsIf {$(\graphval{k}((M^i)_{i < k}, \med, (M^i)_{i > k}) < v_k)$}
		\State{$\low \gets \med$, $\med \gets (\low + \high)/2$, goto~\ref{line:finalloop}}
\EndIf
\end{algorithmic}
\end{algorithm}

\begin{theorem}\label{thm:main}
For ASRG $\varGamma$, $\eqmcost{1}$ returns marginal cost vector $\vec{M}$ so that $\graphflow(\vec{M})$ $ = (\vec{w}, \vec{f})$, where $\vec{w} = \vec{v}$ and $\vec{f}$ is the equilibrium flow in $\varGamma$.
\end{theorem}

The main ingredient in the proof of the theorem is the following lemma, which shows that recursively, for any player $k$, the function $\eqmcost{}$ returns correct marginal costs.

\begin{lemma}
For any vector $(M^1, \dots, M^{k-1})$ with each component in $[0, \Lambda]$, the function $\eqmcost{k}((M^1, $ $\dots, M^{k-1}))$ returns marginal costs $(M^k, \dots, M^n)$ for the remaining players so that, for each player $i \ge k$, $\graphval{i}(M^1, \dots, M^n)$ $ = v_i$.
\end{lemma}

\begin{proof}
The proof is by induction on $n$. In the base case, $k = n$, and the input is the vector $(M^1, \dots, M^{n-1})$ with each component in $[0, \Lambda]$. By Lemma~\ref{lem:existencenplayers}, there exists $\hat{M}$ so that $\graphval{n}(M^1, \dots, M^{n-1}, \hat{M}) = v_n$. We now show that the value $\hat{M}$ can correctly be found by binary search. Initially, the search interval is $[0, \Lambda]$, and by Claim~\ref{lem:existencenplayers}, $\hat{M}$ lies in the search interval. Assume in some iteration the search interval is $[\low, \high]$; $\hat{M}$ lies in the search interval; and that $\graphval{n}(M^1, \dots, M^{n-1}, \med) > v_n$. Since $v_n > 0$, and $v_n$ $=\graphval{n}(M^1, \dots, M^{n-1}, \hat{M})$ $< \graphval{n}(M^1, \dots, M^{n-1}, \med)$ it follows by Lemma~\ref{lem:monotonenplayers} that $\hat{M} < \med$. Hence, $\hat{M}$ lies in the interval $[\low, \med]$, and we can restrict our search to this space, which is exactly how the binary search proceeds. The case when $\graphval{n}(M^1, \dots, M^{n-1}, \med) < v_n$ is similar, and $\hat{M}$ then lies in the interval $[\med, \high]$.

For the inductive step, we are given player $k < n$. We assume that given any input vector $(M^1, \dots, M^{k})$ with each component in $[0, \Lambda]$, $\eqmcost{k+1}$ returns marginal costs $(M^{k+1}, \dots, M^n)$ for the remaining players so that for each of these remaining players $i \ge k+1$, $\graphval{i}(M^1, \dots, M^n) = v_i$. We need to show that given any input marginal costs $(M^1, \dots, M^{k-1})$ for the first $k-1$ players, $\eqmcost{k}$ finds marginal costs $(M^k, \dots, M^n)$ for players $k$ onwards so that the demand returned for these players $i \ge k$ by $\graphval{i}$ is $v_i$. Firstly, by Lemma~\ref{lem:existencenplayers}, choosing $S = [k, \dots, n]$ and $\hat{w}_i = v_i$ for players $i \in S$, there exist marginal costs $(\hat{M}^k, \dots, \hat{M}^n)$ so that for all $i \ge k$, $\graphval{i}((M^i)_{i < k}, (\hat{M}^i)_{i \ge k}) = v_i$. We now show that the binary search procedure in $\eqmcost{k}$ finds the required marginal cost $\hat{M}^k$. By the lemma, $\hat{M}^k$ lies in the initial search interval $[0, \Lambda]$. Assume that in some iteration, $\hat{M}^k$ lies in the search interval $[\low, \high]$, and $\med = (\low + \high)/2$. By the induction hypothesis, $\eqmcost{k+1}(M^1, \dots, M^{k-1}, \med)$ returns marginal costs $(M^{k+1}, \dots, M^n)$ for the players $k+1, \dots, n$ so that for each of these players $i \ge k+1$ (but not player $k$), $\graphval{i}((M^i)_{i<k}, \med, (M^i)_{i>k}) = v_i$. Further, for each player $i \ge k$, $\graphval{i}((M^i)_{i<k},$ $ (\hat{M}^i)_{i\ge k}) = v_i$. Suppose that for player $k$, $\graphval{k}((M^i)_{i<k}, \med, (M^i)_{i>k})$ $ > v_k$ $= \graphval{i}((M^i)_{i<k}, (\hat{M}^i)_{i\ge k})$. Then by Lemma~\ref{lem:monotonenplayers}, $\med > \hat{M}^k$, and hence $\hat{M}^k$ lies in the interval $[\med, \high]$. The algorithm then reduces the search space to this interval, and continues. If $\graphval{k}((M^i)_{i<k}, \med, (M^i)_{i>k})$ $ < v_k$, it can be similarly shown that $\med < \hat{M}^k$. Thus, $\hat{M}^k$ always lies in the search space $[\low, \high]$, which is halved in each iteration. The binary search is thus correct, and must eventually terminate.
\end{proof}

\begin{proof}[Proof of Theorem~\ref{thm:main}]
By the lemma, $\eqmcost{1}$ returns a marginal cost vector $\vec{M}$ $= (M^1, \dots,$ $ M^n)$ so that $\graphval{i}(\vec{M})$ $ = v_i$ for each player $i$. Let $\vec{v} = (v_1, \dots, v_n)$. By definition of the function \graphval{}, this implies that $\graphflow(\vec{M})$ returns vectors $\vec{v}$ and $\vec{f}$. Finally by Claims~\ref{claim:graphfloweq} and~\ref{claim:unique}, $\vec{f}$ is the equilibrium flow for demands $\vec{v}$, as required.
\end{proof}


\subparagraph*{Implementation and Complexity.}
Under the assumption that binary search could be done to arbitrary precision, we showed that the algorithm \eqmcost{} is correct.  However, the solutions to the polynomial equations could be irrationals, and thus the algorithm given is not a finite algorithm. We show in Appendix \ref{sec:implem} that for any given error parameter $\epsilon$, we can implement \eqmcost{} to run in time $O\left(\text{poly}(\log \Psi, \log \frac{1}{\epsilon}, m, n)\right)$, and return an $\epsilon$-equilibria. We say a flow $\vec{f}$ is an $\epsilon$-equilibrium if any player $i$ has flow only on $\epsilon$-minimum marginal cost edges. That is, if $f_e^i > 0$ for player $i$ on edge $e$, then $L_e^i(f) \le \min_{e'} L_{e'}^i(f) + \epsilon$. Conventionally, a strategy profile is an $\epsilon$-equilibrium if no player can improve it's cost by $\epsilon$. One can check that the two are equivalent: if a flow $f$ is an $\epsilon$-equilibrium by our definition, then no player $i$ can improve its cost by more than $\epsilon v_i$, where $v_i$ is its demand.

The work in giving an implementation for the algorithm \eqmcost{} is in implementing the binary search correctly, up to some error parameter $\epsilon$. Since the algorithm is iterative, this error grows across each iteration, and bounding the error in each iteration is quite technical. Additionally, approximate versions of some of the results for \eqmcost{} have to be reproved. For example, Lemma~\ref{lem:monotonenplayers2} and Corollary~\ref{lem:monotonenplayers2} replace Lemma~\ref{lem:monotonenplayers}. The basic framework of our implementation and analysis is similar to our earlier analysis, but differs in many details. We give further details in Appendix \ref{sec:implem}.

\section{An Algorithm with Complexity Exponential in Number of Edges}
\label{sec:typesexp}
Our second algorithm is based on the following theorem, which shows that at equilibrium the supports of players form chains.
\begin{theorem}[\cite{umang}]
\label{thm:chainsupp}
Consider an ASRG with $n$ players on a graph consisting of parallel edges\footnote{The proof by Bhaskar et al.~\cite{umang} is for series-parallel graphs which are a superset of parallel link graphs.}, and let $f$ be the equilibrium flow. Then $L^1(f) \ge \dots \ge L^n(f)$. Consequently, the supports $S_1(f) \supseteq \dots \supseteq S_n(f)$.
\end{theorem}

Thus in an ASRG on $m$ parallel links, there exist numbers $1 = a_1 < a_2 < \dots < a_T \le n$ with $T \le m$, so that players with indices in $[a_i, a_{i+1}-1]$ have the same support at equilibrium. Define a \emph{type set} $\mcT = (P_1, \dots, P_T)$ with $T \le m$ to be a partition of the players so that players in a set $P_t$ in the partition have consecutive indices. Hence, a type set $\mcT = (P_1, \dots, P_T)$ can be denoted by a sequence of numbers $(a_1, \dots, a_T)$ where $1 = a_1 < a_2 < \dots < a_T \le n$ and $P_t$ consists of the players with indices $a_t, \dots, a_{t+1}-1$. We say a type set is \emph{valid} for $\varGamma$ iff two players in the same partition in $\mcT$ also have the same support in the equilibrium. Theorem~\ref{thm:chainsupp} then shows that in a graph consisting of $m$ parallel links, there is a type set that is valid.

We will now give an algorithm with running time that is exponential in the number of edges, using the algorithm from Section~\ref{sec:playersexp}, which is exponential in the number of players, and Theorem~\ref{thm:swamy}. Our algorithm in this section crucially uses Lemma~\ref{lem:typeflow}, which has the following content. Let $\mcT = (P_1, \dots, P_T)$ be a type set, not necessarily valid, for a game $\varGamma$. Consider the game $\varGamma^{\mcT}$ where for each set $P_t \in \mcT$, we replace the players in $P_t$ with $|P_t|$ players that have the same demand, given by $\sum_{i \in P_t} v_i/|P_t|$. That is, we pick a set $P_t$, and replace all players in this set by players with demands equal to the average demand of players in $P_t$. We do this for each set $P_t$. Lemma~\ref{lem:typeflow} then says that if $\mcT$ is valid for $\varGamma$, then the total flow on any edge does not change between $\varGamma$ and $\varGamma^{\mcT}$.

\

\begin{lemma}
Let $\mcT=(P_1, \dots, P_T)$ be the valid type set for game $\varGamma$ with $n$ players on a network of parallel edges. Let $f$ and $g$ be the respective equilibrium flows for games $\varGamma$ and $\varGamma^{\mcT}$ respectively. Then on each edge $e$, $f_e = g_e$.
\label{lem:typeflow}
\end{lemma}

The proof of Lemma~\ref{lem:typeflow} closely follows an earlier proof of uniqueness of  equilibrium in an ASRG~\cite[Theorem 5]{umang}. The lemma implies that in an ASRG on parallel edges we can replace players that have the same support with players that have the same demand without affecting the total flow at equilibrium on the edges. Further, given the total flow on each edge at equilibrium in a general network, the flow for each player can be computed by Theorem~\ref{thm:swamy}. 

\begin{proof}[Proof of Lemma~\ref{lem:typeflow}]
The proof is by contradiction. Let $E^+$ be the set of edges with $f_e > g_e$, and $E^-$ be the remaining edges. Let $P^+$ $:= \{t : \exists e \in E^+ ~ \sum_{i \in P_t} f_e^i > \sum_{i \in P_t} g_e^i\}$ be the set of types that, on any edge in $E^+$, have more flow in $f$ than in $g$. Let type $t \in P^+$, and $e \in E^+$ be an edge so that $\sum_{i \in P_t} f_e^i > \sum_{i \in P_t} g_e^i\}$. Since all players in $P_t$ have the same support, this implies that $f_e^i > 0$ for all players $i \in P_t$. Further, this implies that

\begin{align}
\sum_{i \in P_t} L^i(f) ~=~ \sum_{i \in P_t} L_e^i(f) ~>~ \sum_{i \in P_t} L_e^i(g) \ge \sum_{i \in P_t} L^i(g) \, .
\label{eqn:typeflow}
\end{align}

Let $P^-$ be the remaining player types, and note that for any type $t \in P^-$, $\sum_{i \in P_t} f_e^i \le \sum_{i \in P_t} g_e^i\}$. Now consider the total difference in flow on edges in $E^+$:

\[
0 ~<~ \sum_{e \in E^+} f_e - g_e ~=~ \sum_{e \in E^+} \sum_{t \in P^+} \sum_{i \in P_t} \left( f_e^i - g_e^i \right) + \sum_{e \in E^+} \sum_{t \in P^-} \sum_{i \in P_t} \left( f_e^i - g_e^i \right).
\]

\noindent The first inequality is by definition of $E^+$. Now note that the last summand is non positive, hence
\[
\sum_{e \in E^+} \sum_{t \in P^+} \sum_{i \in P_t} \left( f_e^i - g_e^i \right) > 0
\]

\noindent Hence there is a player type $t \in P^+$ for which the total flow on edges in $E^+$ is strictly greater in $f$ than in $g$. Since the total flow across all edges for player type $t$ must remain conserved, there is an edge in $E^-$ with $\sum_{i \in P_t} f_e^i < \sum_{i \in P_t} g_e^i$. But this gives us that $\sum_{i \in P_t} L^i(f) < \sum_{i \in P_t} L^i(g)$, contradiction~\eqref{eqn:typeflow}.
\end{proof}

Our algorithm for computing equilibrium in a game $\varGamma$ is now as follows. We first enumerate over all type sets with at most $m$ types. For each type set $\mcT = (P_1, \dots, P_T)$, we consider the game $\varGamma^{\mcT}$, and modify the algorithm \eqmcost{} from the previous section to return the equilibrium $f^{\mcT}$ for $\varGamma^{\mcT}$, with running time exponential in $T \le m$. We then use Theorem~\ref{thm:swamy} to check if the total flow on each edge in flow $f^{\mcT}$ can be decomposed into an equilibrium flow for $\varGamma$. By Theorem~\ref{thm:chainsupp}, there is a valid type set $\mcT$, and then by Lemma~\ref{lem:typeflow} for the valid type set $\mcT$, the flow on each edge in $f^{\mcT}$ is equal to the flow on each edge at equilibrium in $\varGamma$. Theorem~\ref{thm:swamy} then gives us a decomposition into flows from each player.

In Algorithm \eqmcostty{}, for simplicity, we use the notation $M^i_{\times |P_i|}$ to denote the vector where $M^i$ is repeated $|P_i|$ times. Further, $(M^i_{\times |P_i|})_{i > t}$ is used to denote the vector
\[
\left(\underbrace{M^{t+1}, \ldots, M^{t+1}}_{|P_{t+1}| \text{ times}}, \ldots, \underbrace{M^T, \ldots, M^T}_{|P_T| \text{ times}} \right) \, .
\]

\begin{algorithm}[!h]
\caption{\eqmcostty{t}($\mcT=(P_1, \dots, P_T)$, $(M^1, \dots, M^{t-1})$)} \label{algo:eqmcostty}
\begin{algorithmic}[1]
\Require{ Vector $(M^1, \dots, M^{t-1})$, with each component $M^j \in [0,\Lambda]$. Type set $\mcT=(P_1, \dots, P_T)$} \Comment{If $t=1$, only the type set is required as input.}
\Ensure{Vector $(\overrightarrow{M})$ of marginal costs so that the first $a_{k-1}$ marginal costs are equal to the inputs, and for players $i \ge a_{k-1}+1$, the demand $\graphval{i}(\vec{M}) = v_i$.}
\If{$t = T$}
	\State {Using binary search in $[0,\Lambda]$, find $M$ so that $\graphval{n}((M^i_{\times |P_i|})_{i < T}, \underbrace{M,M,\ldots,M}_{|P_T| \text{ times}} ) = v_n$.}
	\State \Return {$M$}
\EndIf
\State {$\low \gets 0$, $\high \gets \Lambda$, $\med \gets (\low + \high)/2$.}
\State {$(M^{t+1}, \dots, M^n) \gets \eqmcost{t+1}(\mcT, (M^1, \dots, M^{t-1}, \underbrace{\med,\med,\ldots,\med}_{|P_t| \text{ times}})$} \label{line:finalloop4} \Comment{Call \eqmcostty{t+1} to get marginal costs for the remaining player types $t+1, \dots, T$ so that the demand for these player types is correct.}
\If {$(\graphval{a_t}((M^i_{\times |P_i|})_{i < t}, \underbrace{\med,\med,\ldots,\med}_{|P_t| \text{ times}}, ((M^i_{\times |P_i|})_{i > t}) = v_{a_t})$}
		\State \Return {$(\med, (M^i)_{i > t})$}
\ElsIf {$(\graphval{a_t}((M^i_{\times |P_i|})_{i < t}, \underbrace{\med,\med,\ldots,\med}_{|P_t| \text{ times}}, (M^i_{\times |P_i|})_{i > t}) > v_{a_t})$}
		\State{$\high \gets \med$, $\med \gets (\low + \high)/2$, goto~\ref{line:finalloop4}}
\ElsIf {$(\graphval{a_t}((M^i_{\times |P_i|})_{i < t}, \underbrace{\med,\med,\ldots,\med}_{|P_t| \text{ times}}, (M^i_{\times |P_i|})_{i > t}) < v_{a_t})$}
		\State{$\low \gets \med$, $\med \gets (\low + \high)/2$, goto~\ref{line:finalloop4}}
\EndIf
\end{algorithmic}
\end{algorithm}

\begin{claim}
 Consider ASRG on parallel edge graph with $m$ edges and $n$ players. Then number of possible typesets at equilibrium is
 \[\binom{n+m-1}{m-1}.\]
\label{claim:typecount}
\end{claim}
\begin{proof}
Since supports of players at equilibrium form chains (Theorem \ref{thm:chainsupp}), we have $S_1\supseteq S_2 \supseteq \ldots \supseteq S_n$. Consider an equilibrium flow $f$  and let $(P_1,P_2,\ldots,P_m)$ is the typeset, where each $P_i\subseteq[n]$. By Theorem \ref{thm:chainsupp}, we can rename $P_1,P_2,\ldots,P_m$ such that for every pair of players $a$ and $b$, if $a\in P_i$ and $b\in P_j$ such that $i<j$ then $v_a\geq v_b$. Therefore, it is easy to see that number of ways of choosing the typeset $(P_1,P_2,\ldots,P_m)$ is same as the number of ways one can fill $m$ bins with $n$ identical balls(some bins can be empty). By simple counting we get this number to be $\binom{n+m-1}{m-1}.$
\end{proof}

Here, we briefly describe the modification to \eqmcost{} so that it runs in time exponential in the number of edges, rather than the number of players $n$. 

The algorithm \eqmcost{} runs a binary search for the marginal cost at equilibrium for each player. For each player $i$, \eqmcost{i} runs a binary search, in each iteration of which it calls \eqmcost{i+1}. If each binary search runs for $R$ iterations, and each iteration takes time $S$, then this recursively gives a running time of $O(S\,R^n)$. However, if players $i$, $i+1$ have the same demand, then by Theorem~\ref{thm:chainsupp} they also have the same marginal cost at equilibrium. Hence we can run the binary search for the marginal cost at equilibrium, for both players \emph{simultaneously}. This is the basic idea for the modification. For a given type set $\mcT = (P_1, \dots, P_T)$, and the game $\varGamma^{\mcT}$, we know that all players in the same set $P_t$ have the same marginal cost at equilibrium. Hence, we run the binary search once for each set $P_t$, rather than once for each player. By Theorem~\ref{thm:chainsupp}, the number of types $T \le m$.

This gives us the reduced running time of $O(S\,R^m)$ for type set $\mcT$, for each run of the modified algorithm \eqmcostty{}. However, note that we run the algorithm once for each possible type set. If the game is played on $m$ parallel links, then Claim~\ref{claim:typecount} shows that the number of possible type sets is about $(n+m)^m$. From Theorem~\ref{thm:implement}, this gives us a running time of $O\left((n+m)^mm^3 \left(\log (n\Psi/\epsilon) \right)^m\right)$.

\section{Hardness of Computing Equilibria}
\label{sec:hardness}

Prior to the proof of PPAD-hardness of computing a (mixed) Nash equilibrium in bimatrix games, Gilboa and Zemel showed that it was \nphard to determine if there existed an equilibrium in bimatrix games where the every player had payoff above a given threshold $C$~\cite{gilboa}. We show a similar result for ASRGs.

\begin{theorem}\label{thm:hard}
Given an ASRG with convex, strictly increasing and continuously differentiable cost functions, it is \textsf{NP-hard} to determine if there exists an equilibrium at which cost of every player is at most $C$.
\end{theorem}

The main idea of the proof is to build upon the existence of multiple equilibria in ASRGs, and is a reduction from \subsum. In \subsum, we are given a set $S=\{s_1,s_2,\ldots,s_n\}\subset \mathbb{N}$ such that the sum of elements in $S$ is $M$, and we want to determine if there exists a subset $T\subseteq S$ such that the sum of elements in $T$ is $M/2$. This problem is known to be \textsf{NP}-complete. Our reduction is in two steps. First, we construct an ASRG $\mcG$ with four players $b$, $r$, $p$, and $q$, and exactly three equilibria, one of which is irrational, which is used as gadget in the reduction (Figure~\ref{fig:mul_eq_game_3players}).\footnote{We use Mathematica to verify properties of equilibria in the games used in the reduction. This is explained in the appendix, and the files used are available on the second author's home page.} This construction builds upon an earlier example showing multiplicity of equilibria in ASRGs~\cite{umang}. We will be mainly concerned with the rational equilibrium flows, say $f$ and $g$. We choose the cost functions so that (i) $\mathcal{C}_{e_7}^p(g) > \mathcal{C}_{e_7}^p(f)$, and (ii) the sum of costs of players $p$ and $q$ are equal for $f$ and $g$, that is, 
\begin{equation}
\Lambda ~:=~ \mathcal{C}_{e_7}^p(f) + \mathcal{C}_{e_9}^q(f) ~=~ \mathcal{C}_{e_7}^p(g) + \mathcal{C}_{e_9}^q(g)
\label{eqn:C}
\end{equation}

\noindent Then $\mathcal{C}_{e_9}^q(f) > \mathcal{C}_{e_9}^q(g)$.

We now repeat this subgame $n$ times in series, once for each element in the set $S$ in the \subsum instance (see Figure~\ref{fig:reduction2}). Each subgame is independent of the others, i.e., the players $b_i$ and $r_i$ in $i^{th}$ subgame are local to that subgame and do not play any role in other subgames. All the subgames are connected by players $p$ and $q$, who can only use one edge ($e_7$ and $e_9$ respectively) in each subgame. We show that $f$, $g$, and $h$ continue to be the only equilibria within each subgame. In the $i$th subgame $\mcG_i$, we multiply each cost function by $s_i$. This causes all costs to get multiplied by $s_i$, and does not affect the equilibria. Thus, in each subgame $\mcG_i$, player $p$ has costs $s_i \mathcal{C}_{e_7}^p(f)$ and $s_i \mathcal{C}_{e_9}^p(g)$ in equilibrium flows $f$, $g$, and $h$ (confined to the subgame) respectively. Similar for player $q$. Roughly, we think of equilibrium $f$ in a subgame as putting $s_i$ in the subset $S$, and equilibrium $g$ as leaving $s_i$ out. 

We will show that if the given instance satisfies \textsf{SUBSET-SUM}, then there exists an equilibrium at which both players $p$ and $q$ have cost $(M \Lambda)/2$, otherwise at least one of them has cost strictly greater than $(M \Lambda)/2$. At equilibrium in the game, let $F$ be the subgames where $f$ is the equilibrium, and $G$ be the subgames where $g$ is the equilibrium. Then the total cost of players $p$ and $q$ is
\[
\mathcal{C}_{e_7}^p(f) \sum_{i \in F} s_i + \mathcal{C}_{e_7}^p(g) \sum_{i \in G} s_i + \mathcal{C}_{e_9}^q(f) \sum_{i \in F} s_i + \mathcal{C}_{e_9}^q(g) \sum_{i \in G} s_i  ~=~  M \Lambda
\]
\noindent where the equality follows from~\eqref{eqn:C}. From $\mathcal{C}_{e_7}^p(g) > \mathcal{C}_{e_7}^p(f)$, it follows that the cost of each player $p$, $q$ is $(M \Lambda)/2$ iff at equilibrium, $\sum_{i \in F} s_i = \sum_{i \in G} s_i$. Else, since the sum of costs of the two players is constant, exactly one player has cost above $(M \Lambda)/2$. 

To complete the proof, we add player-specific edges to ensure that $(M \Lambda)/2$ is large enough so that all the other players $b_i, r_i$ always have cost at most $(M \Lambda)/2$ at any equilibrium. 

In the appendix, we give the detailed reduction, as well as proofs of the properties of equilibria that we use. The exact calculations are done using Mathematica, the files for which are available at the second author's homepage. Here we give a high-level description of the reduction and the main points of the proof.

\subsection{Gadget \texorpdfstring{$ \mathcal{G} $}{G}  }
As mentioned earlier, game $ \mcG $ has $ 4 $ players $ b$, $r$, $p$, and $q $ with network shown in Figure \ref{fig:mul_eq_game_3players}. Players $ b $ and $ r $ want   to send $v_{b}=4763.5$ and $v_{r}=2415.3$ units of flow from $s$ to $t$. Players $p$ and $q$ want to send $v_p=100$ and $v_q=100$ units of flow from $s_p$ to $t_p$ and $s_q$ to $t_q$ respectively. 

\begin{figure}[!ht]
\centering
\includegraphics[scale=0.8]{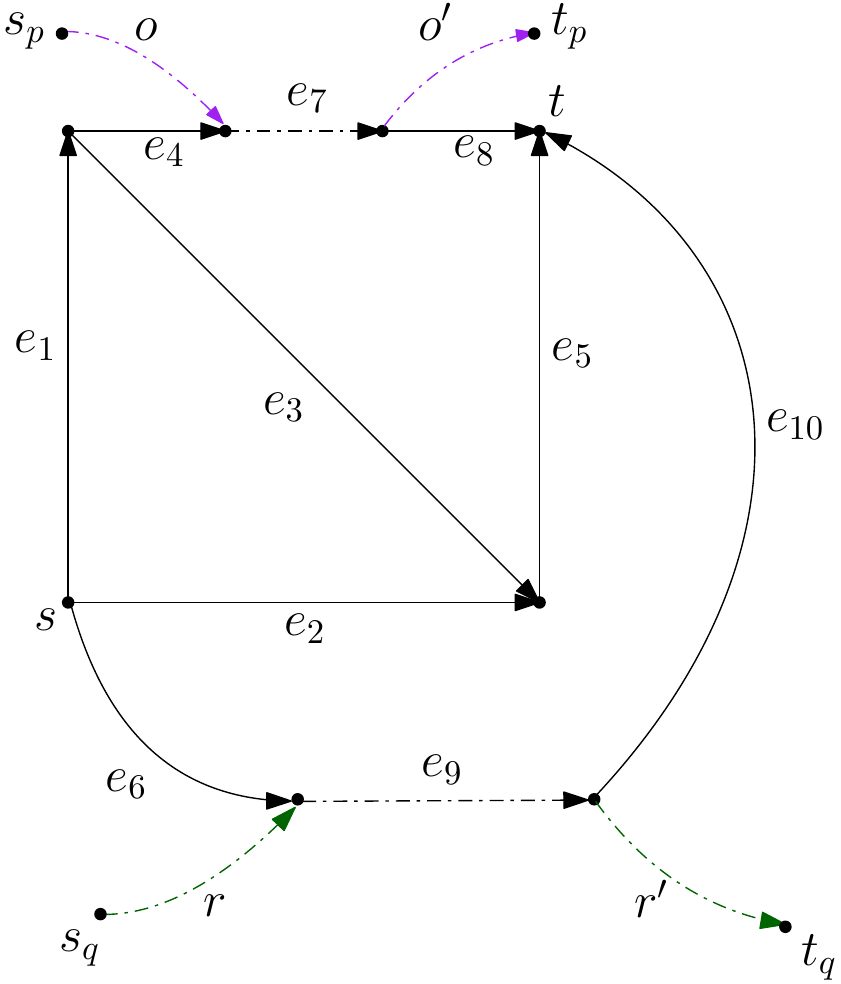}
\caption{ASRG $\mathcal{G}$, with multiple equilibria}
 \label{fig:mul_eq_game_3players}
\end{figure}

We will show in the Appendix (Lemma~\ref{lem:mul_eq_T}) that $ \mcG $ has $ 3 $ equilibria $ f,g $ and $ h $ where $ f$ and $ g $ are rational equilibria i.e., have rational flows on edges and $ h $ is an irrational equilibrium. The approximate equilibrium flow values are given in Table~\ref{tab:eq_flows_main}.

\begin{table}[!ht]  
\centering
\caption{Equilibrium flows for ASRG $\mcG$ accurate up to $2$ decimals}
\begin{tabular}{|c|c|c|c|c|c|c|}
\hline
\multirow{2}{*}{Edge} & \multicolumn{2}{c|}{$f$} &\multicolumn{2}{c|}{g} &\multicolumn{2}{c|}{h}\\
\cline{2-7}
& $b$ & $r$ &  $b$ & $r$ & $b$ & $r$
\\ \hline
  e1 & $500$ & $100$ &  $540$ & $50$ & $527.41$ & $71.78$ \\
  e2 & $500$ & $0$ & $540$ &$0$ & $527.41$ &$0$ \\
  e3 & $0$ & $100$ & $0$ & $50$ & $0$&$71.78$  \\
  e4 & $500$ & $0$ & $540$ & $0$ & $527.41$ &$0$ \\
  e5 & $500$ & $100$ & $540$ & $50$ & $527.41$ & $71.78$ \\
  e6 & $3763.5$ & $2315.30$ & $3683.5$ & $2365.30$ &  $3708.69$ &$2343.54$ \\
  e7 & $500$ & $0$ & $540$ & $0$ & $527.41$ &$0$ \\
  e8 & $500$ & $0$ & $540$ & $0$ & $527.41$ &$0$ \\
  e9 & $3763.5$ & $2315.30$ & $3683.5$ & $2365.30$ &  $3708.69$ &$2343.54$ \\
  e10 & $3763.5$ & $2315.30$ & $3683.5$ & $2365.30$ &  $3708.69$ &$2343.54$ \\

  \hline
\end{tabular}\label{tab:eq_flows_main}
\end{table}
\begin{proof}[Proof of Theorem~\ref{thm:hard}]
Let $\mathcal{C}^p(f)$, $\mathcal{C}^p(g)$ and $\mathcal{C}^p(h)$ are the costs of player $p$ at equilibria $f$, $g$ and $h$ respectively in game $\mcG$. Similarly $\mathcal{C}^q(f)$, $\mathcal{C}^q(g)$ and $\mathcal{C}^q(h)$ are the costs of player $q$ at equilibria $f$, $g$ and $h$ respectively.

\begin{figure}[!ht]
\centering
\includegraphics[scale=1]{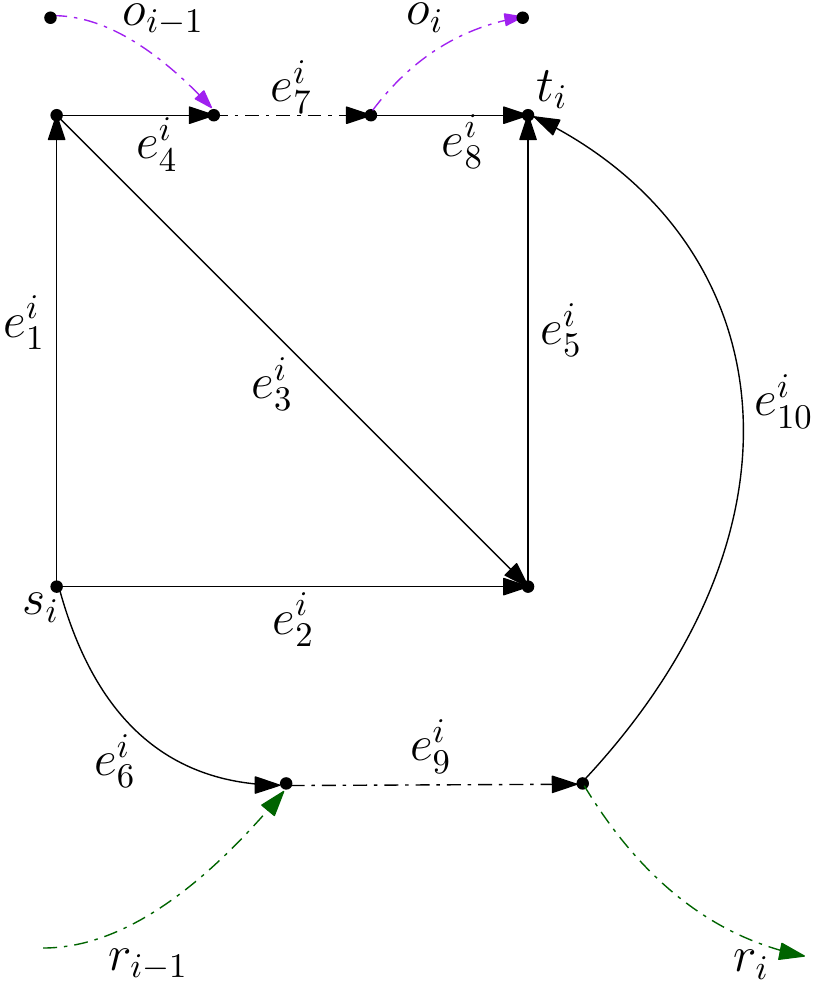}
\caption{Gadget $s_i\mathcal{G}$}
 \label{fig:mul_eq_gadget}
\end{figure}

\begin{figure}[!ht]
\centering
\includegraphics[width=\textwidth]{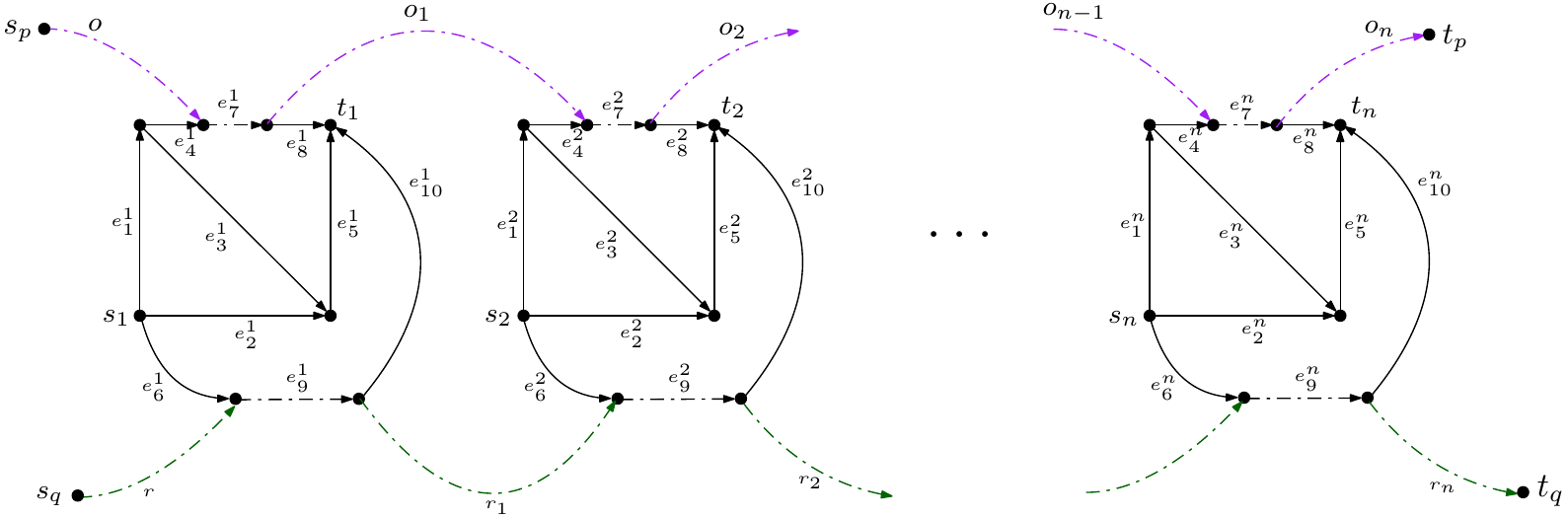}
 \caption{Reduction to \textsf{SUBSET-SUM} problem}
 \label{fig:reduction2}
\end{figure}

Consider ASRG game $\mathcal{H}$ with $2n+2$ players with network as shown in Figure \ref{fig:reduction2}. Players $p$ and $q$ want to send $v_p=100$ and $v_q=100$ units of flow from $s_p$ to $t_p$ and $s_q$ to $t_q$ respectively. For each $i\in [n]$, players $b_i$ and $r_i$ want to send $v_{b_i}=4763.5$ and $v_{r_i}=2415.3$ units of flow  from $s_i$ to $t_i.$

The cost functions for edges $e_j^i$ is $l_{e_j^i}(x)=s_il_{e_j}(x)$ $\forall j\in [10], i\in [n]$. The cost functions on edges $o,o_i$ is  $l_{o_i}(x)=x$ and for edges $r,r_i$ is $l_{r_i}(x)=x$ $\forall i\in [n]$. Observe the following:
\begin{itemize}
  \item Player $p$ has only one path to send his flow i.e., via $o-e_{7}^1-o_1-e_7^2-o_2-\ldots-o_n.$ Similarly player $q$ has only one path to send his flow i.e., via $r-e_9^1-r_1-e_9^2-r_2-\ldots-r_n.$
  \item For each $i$, players $b_i$ and $r_i$ want to send their flow from $s_i$ to $t_i$. They can only use the edges $e_{j}^i$, $j\in [10]$. Since the costs of edges $e_{j}^i$ is $l_{e_j^i}(x)=s_il_{e_j}(x)$, $\forall j\in [10]$, the equilibrium conditions of $b_i$ and $r_i$ in $\mcG_i$ (and so in $\mathcal{H}$) are same as the equilibrium conditions of $b$ and $r$ in $\mathcal{G}$, when multiplied by $s_i$ on both sides. Therefore, the equilibria of $b_i$ and $r_i$ in $\mcG_i$ (and so in $\mathcal{H}$) are the same as the equilibria of $b$ and $r$ in $\mathcal{G}$.
 \end{itemize}

In Claim \ref{claim:cost_sum} in the appendix, we show that $C_{e_7}^p(f)+C_{e_9}^q(f)=C_{e_7}^p(g)+C_{e_9}^q(g)$ i.e., the sum of costs of players $p$ and $q$ in $\mcG$ is some constant $A$. Since the cost functions in each gadget $\mcG_i$ are multiples of cost functions in $\mathcal{G}$, the above claim implies that the sum of costs by players $p$ and $q$ on $\mcG_i$ is the $s_iA$ if the flow is either $f$ or $g$ in $\mcG_i$. This allows us to make the following claim. Let $CP= M/2(C_{e_7}^p(g)+C_{e_7}^p(f))+10^4(n+1)$ and $CQ=M/2(C_{e_9}^q(f)+C_{e_9}^q(g))+10^4(n+1)$. In the following lemma, we show that $CP$ and $CQ$ are the costs of players $p$ and $q $ respectively when the equilibrium in game $\mathcal{H}$ is such that if gadget $i$ has flow $f$, then $s_i\in T$, and if it has flow $g$, then $s_i\notin T.$

\begin{lemma}
 Suppose $\mathcal{G}$ has only two rational equilibria, $f$ and $g$. Then in ASRG $\mathcal{H}$, there exists an equilibrium such that the players $p$ and $q$ have costs $CP$ and $CQ$ respectively iff set $S$ satisfies the \textsf{SUBSET-SUM} problem. If $S$ does not satisfy \subsum then at every equilibrium, either player $p$ has cost stricly greater than $CP$ or player $q$ has cost stricly greater than $CQ.$
\begin{proof}
  Let $S$ satisfy \textsf{SUBSET-SUM} and $T, T'$ be the satisfying partition. Observe that the gadget $s_i \mathcal{G}$ corresponds to the element $s_i$ in the set $S$. Consider the equilibrium of $\mathcal{H}$ where if $s_i\in T$, then gadget $s_i\mathcal{G}$ has equilibrium flow $f$ else equilibrium flow $g.$ Therefore, the cost of player $p$ at this equilibrium in $\mathcal{H}$ is
  \begin{align*}
    &= \sum_{s_i \in T }s_i C_{e_7}^p(f)+ \sum_{s_i \in T' }s_i C_{e_7}^p(g)+10^4(n+1)\\
    &=M/2(C_{e_7}^p(g))+M/2(C_{e_7}^p(f))+10^4(n+1)\\
    &=M/2(C_{e_7}^p(g)+C_{e_7}^p(f))+10^4(n+1).
  \end{align*}
Similarly, one can argue for player $q$. 

Now consider the case  where $S$ does not satisfy \subsum. Let $T, T'$ be a partition of set $S$, such the if $s_i\in T$, then gadget $s_i\mathcal{G}$, has equilibrium $f$ else has equilibrium $g.$ Now consider the sum of costs of players $p$ and $q$ at this equilibrium in $\mathcal{H}$, which is
  \begin{align*}
    &=\sum_{s_i \in T }s_i C_{e_7}^p(f)+ \sum_{s_i \in T' }s_i C_{e_7}^p(g)+10^4(n+1)\\
   & \qquad \qquad\qquad+\sum_{s_i \in T }s_i C_{e_9}^q(f)+ \sum_{s_i \in T' }s_i C_{e_9}^q(g)+10^4(n+1)\\
    &=\sum_{s_i \in T }s_i (C_{e_7}^p(f)+C_{e_9}^q(f))+\sum_{s_i \in T' }s_i(C_{e_7}^p(g)+C_{e_9}^q(g))+2.10^4(n+1)
  \end{align*}
By Claim \ref{claim:cost_sum}, we can write the above expression as
\begin{align*}
  &=M/2(C_{e_7}^p(f)+C_{e_9}^q(f))+M/2(C_{e_7}^p(g)+C_{e_9}^q(g))+2.10^4(n+1)\\
  &=\left\{M/2(C_{e_7}^p(g)+C_{e_7}^p(f))+10^4(n+1)\right\}+\left\{M/2(C_{e_9}^q(f)+C_{e_9}^q(g))+10^4(n+1)\right\}\\
  &=CP+CQ.
\end{align*}
Therefore at any equilibrium, sum of costs of players $p$ and $q$ is $CP+CQ$ which is constant. Suppose partition  $T,T'$ is such that sum of elements in $T$ is $>M/2$. Then the cost of player $p$ at the corresponding equilibrium is
\begin{align*}
  &=\sum_{s_i \in T }s_i C_{e_7}^p(f)+ \sum_{s_i \in T' }s_i C_{e_7}^p(g)+10^4(n+1)\\
   &< M/2(C_{e_7}^p(f)+C_{e_7}^p(g))+10^4(n+1) \qquad \qquad ( C_{e_7}^p(f)<C_{e_7}^p(g).)
\end{align*}

Since, the sum of costs of players $p$ and $q$ is constant the cost of player $q$ at this equilibrium flow is $>CQ.$
\end{proof}
\end{lemma}

It can be calculated that $C^p_{e_7}(h)+C_{e_9}^q(h)>C_{e_7}^p(f)+C_{e_9}^q(f)=C_{e_7}^p(g)+C_{e_9}^q(g)$ (detailed calculations are done in provided Mathematica file named "Flows.nb"). It is easy to see that if any gadget has equilibrium flow $h$, then the sum of the costs of players $p$ and $q$ is $>CP+CQ$. Hence, either player $p$ has cost $>CP$ or player $q$ has cost $>CQ.$

Hence, we have shown that $S$ satisfies \subsum, iff at some equilibrium the players $p$ and $q$ has cost $CP$ and $CQ$ respectively. Define
\[C_1=\max_{i\in [n]}\left\{\max_{\substack{\text{equilibrium} \\ \text{flows } f,g,h}}\left\{\text{cost of player }b_i, \text{cost of player }r_i\right\}\right\}\]
and $C=\max\{C_1,CP,CQ\}$. Then extra edges can be added to the paths of players $p$ and $q$ respectively, such that both $p$ and $q$ have cost $C$ iff $S$ satisfies \subsum. Also, at any possible equilibrium flow of $\mathcal{H}$, for each $i\in [n]$, cost of players $b_i$ and $r_i$ is at most $C$. This proves the theorem.

\end{proof}

\bibliography{bib_agt}

\appendix
\section{Appendix} 
\label{sec:appendix}




\subsection{Implementation}
\label{sec:implem}
We start with an implementation of the binary search procedure to find an approximate solution to a polynomial equation.

\begin{algorithm}[!ht]
\caption{\binsearchh($k$, $e$, $M$, $\delta$)}\label{algo:binsearchh}
\begin{algorithmic}[1]
\Require{Integer $k$, edge $e$, function value $M \ge k \, l_e(0)$, precision $\delta > 0$.}
\Ensure{Nonnegative flow $\hat{x}$ so that $|\hat{x} - x^*| \le \delta$, where $x^*$ solves $k l_e(x) + xl_e'(x) = M$.}
\State {$\low \gets 0$, $\high \gets M \Psi$}
\While {$\high - \low \ge \delta$}
	\State {$\med \gets (\high - \low)/2$}
	\If {$k l_e(\med) + \med \, l_e'(\med) > M + \delta/(2 \Psi)$}
		\State {$\high \gets \med$}
	\Else \If {$k l_e(\med) + \med \, l_e'(\med) < M - \delta/(2 \Psi)$}
		\State {$\low \gets \med$}
		\Else \State \Return {$\med$}
		\EndIf
	\EndIf
\EndWhile
\end{algorithmic}
\end{algorithm}

We note that from Algorithm \graphflowh, the highest value of $M$ that \binsearchh is called with is $n \Lambda$ $\le 2n \Psi^2$.

\begin{lemma}If $M \ge k l_e(0)$, the algorithm \binsearchh returns value $\hat{x}$ so that $|\hat{x} - x^*| \le \delta$, where $x^*$ solves $k l_e(x) + xl_e'(x) = M$, in time $O\left(\log \frac{2n \Psi^3}{\delta}\right)$.
\label{lem:binsearchh}
\end{lemma}

\begin{proof}
It is easy to see that the algorithm terminates in the stated time, since in each iteration of the while loop, the algorithm either terminates and returns a value $\hat{x}$ so that $k l_e(\hat{x}) + \hat{x} l_e'(\hat{x})$ $\in [M \pm \delta/(2 \Psi)$, or halves the difference $\high - \low$. In any case, it terminates if $\high - \low \le \delta$, giving us the bound on the time complexity.

To see that $\hat{x}$ has the required property, note that for any $x$, if $|x - x^*| \ge \delta$, then

\[
\left\| \left(k l_e(x) + x l_e'(x)\right) - \left( k l_e(x^*) + x^* l_e'(x^*) \right)\right\| ~ \ge ~ \delta / \Psi \, .
\]

\noindent hence in the first case, if $k l_e(\hat{x}) + \hat{x} l_e'(\hat{x})$ $\in [M \pm \delta/(2 \Psi)$, then $|\hat{x} - x^*| \le \delta$.

Now suppose the algorithm terminates with $\high - \low \le \delta$. We will show that $x^* \in [\low, \high]$, completing the proof since the algorithm returns $\med = (\high - \low)/2$. Consider the first iteration of the while loop. Since $k l_e(0) \ge M$, $x^* \ge 0$. Further, $k l_e(M \Psi) \ge k M$, and hence $x^* \le M \Psi$. Hence in the first iteration, $x^* \in [\low, \high]$. Now suppose that in some iteration the statement is true. Clearly, if $k l_e(\med) + \med \, l_e'(\med)$ $> M$, then $\med > x^*$, and $x^*$ must then be in the interval $[\low, \med]$. The algorithm sets $\high$ to $\med$, and hence the statement is true in the next interval as well. In the other case, if $k l_e(\med) + \med \, l_e'(\med)$ $< M$, then $\med < x^*$, and again the statement that $x^* \in [\low, \high]$ can be verified to be true in the following interval as well. Thus, $x^*$ always lies in the interval $[\low, \high]$, as required.
\end{proof}

The next algorithm is a redistribution procedure for vectors. Given a vector $\vec{M} = (M^1, \dots, M^k)$ and a scalar $\hat{M}$ so that $|\hat{M} - \sum_i M^i| \le \epsilon$, the procedure returns a nonnegative that is component-wise $\epsilon$-close to $\vec{M}$, and sums to $\hat{M}$.

\begin{algorithm}[!ht]
\caption{\redistribh($M$, $\hat{M}$, $k$, $(M^i)_{i \le k}$)}\label{algo:redistribh}
\begin{algorithmic}[1]
\Require{Nonnegative scalars $M$, $\hat{M} \ge 0$, nonnegative $k$-vector $(M^i)_{i \le k}$ of nonnegative real values with $\sum_{i \le k} M^i = M$.}
\Ensure{Nonnegative $k$-vector $(\hat{M}^i)_{i \le k}$ so that $|M^i - \hat{M}^i| \le |M - \hat{M}|$ for all $i \le k$, and $\sum_{i \le k} \hat{M}^i = \hat{M}$.}
\If {$\hat{M} \ge M$} \Comment{In this case, give all the excess to player 1}
	\State {$\hat{M}^1 = M^1 + (\hat{M} - M)$, $\hat{M}^i = M^i$ for all $i \in \{2, \dots, k\}$}
\Else \Comment{If $M > \hat{M}$, distribute deficit starting from player 1, maintaining nonnegativity}
	\State {$R(0) \gets M$, $\hat{R}(0) \gets \hat{M}$, $\delta(0) \gets R(0) - \hat{R}(0)$}
	\For {$i = 1 \to k$}
		\State {$\hat{M}^i = \max \{0, M^i - \delta(i-1)\}$}
		\State {$R(i) \gets R(i-1) - M^i$, $\hat{R}(i) = \hat{R}(i-1) - \hat{M}^i$, $\delta(i) \gets R(i) - \hat{R}(i)$}
	\EndFor
\EndIf
\State \Return {$(\hat{M}^i)_{i \le k}$}
\end{algorithmic}
\end{algorithm}

\begin{lemma}
The algorithm $\redistribh$ returns a vector $(\hat{M}^i)_{i \le k}$ that satisfies $|M^i - \hat{M}^i| \le |M - \hat{M}|$ for all $i \le k$, and $\sum_{i \le k} \hat{M}^i = \hat{M}$.
\label{lem:redistribh}
\end{lemma}

\begin{proof}
In the first case, when $\hat{M} \ge M$, the two properties required in the lemma can easily be verified. In the other case, when $M \ge \hat{M}$, we note that for $i \ge 1$,

\begin{equation}
\delta(i) = R(i) - \hat{R}(i) = R(i-1) - M^i - \hat{R}(i-1) - \hat{M}^i = \delta(i-1) - (M^i - \hat{M}^i) \, . \label{eqn:redistribh}
\end{equation}

\noindent We first show that the deficits $\delta(i)$ are nonincreasing and nonnegative, thus

\begin{equation}
M - \hat{M} = \delta(0) \ge \delta(1) \ge \dots \ge \delta(k) \ge 0 \, .
\label{eqn:redistribh3}
\end{equation}

Intuitively, this follows because each $M^i - \hat{M}^i$ absorbs some (if not all) of the deficit. Concretely, by induction, $\hat{M}^1 = \max\{0, M^1 - \delta(0)\}$. Since $\delta(0) > 0$ and $M^1 \ge 0$, we get that $M^1 \ge \hat{M}^1 \ge M^1 - \delta(0)$. Hence from~\eqref{eqn:redistribh}, $\delta(1) \le \delta(0)$, and $\delta(1) \ge 0$. Now suppose that $\delta(0) \ge \dots \ge \delta(i-1) \ge 0$. Since $\hat{M}^i = \max\{0, M^i - \delta(i-1)\}$, we again obtain that

\begin{equation}
M^i \ge \hat{M}^i \ge M^i - \delta(i-1)
\label{eqn:redistribh2}
\end{equation}

\noindent and hence, from~\eqref{eqn:redistribh}, $\delta(i) \le \delta(i-1)$, and $\delta(i) \ge 0$. Note that the first part of the lemma follows from~\eqref{eqn:redistribh2} and the fact just proven that $\delta(i) \le \delta(0) = M - \hat{M}$.

From~\eqref{eqn:redistribh}, summing over all players $i \le k$,

\[
\sum_{i \le k} (M^i - \hat{M}^i) ~=~ \sum_{i \le k} \delta(i-1) - \delta(i) ~=~ \delta(0) - \delta(k) ~=~ M - \hat{M} - \delta(k) \, .
\]

\noindent We now show that $\delta(k) = 0$, and hence $\sum_{i \le k} \hat{M}^i = \hat{M}$. For a contradiction, suppose that $\delta(k) = 0$. Then from~\eqref{eqn:redistribh3}, each $\delta(i) > 0$, for $i \le k$. From~\eqref{eqn:redistribh}, this implies that $\hat{M}^i > M^i - \delta(i-1)$. By definition, $\hat{M}^i = \max \{0, M^i - \delta(i-1)\}$, hence this implies that $\hat{M}^i = 0$ for each $i \le k$. Thus, $\sum_{i \le k} \hat{M}^i = 0$. Plugging this into the previous displayed equation gives us that $M = M - \hat{M} - \delta(k)$, or $\delta(k) + \hat{M} = 0$. Since $\hat{M}$ is nonnegative and $\delta(k)$ is strictly positive by assumption, this gives us a contradiction. Thus, $\delta(k) = 0$, and $\sum_{i \le k} \hat{M}^i = \hat{M}$, proving the second part of the lemma.
\end{proof}

Algorithm \graphflowh implements \graphflow.

\begin{algorithm}[!ht]
\caption{\graphflowh($\vec{M}, \delta$)}\label{algo:graphflowh}
\begin{algorithmic}[1]
\Require{Vector $\vec{M}=(M^i)_{i \in [n]}$ of nonnegative real values in $[0, \Lambda]$, precision $\delta > 0$}
\Ensure{Flow $\vec{f}$ and demands $\vec{w}$ so that  $\vec{f}$ is an $4n \Psi^2 \delta$-approximate equilibrium flow for demands $\vec{w}$, and $|w_i - \graphval{i}(\vec{M})| \le 4m \delta \Psi^4$.}
\State Assume that $M^1 \ge M^2 \ge \dots \ge M^n$, else renumber the vector components so that this holds.
\For {each edge $e \in E$} \label{line:graphflow2}
	\State {$f_e^i = 0$ for each player $i \in [n]$}
	\If {$l_e(0) + 2n \Psi^2 \delta \ge M^1$}
		\State {$S_e \gets \emptyset$; continue with the next edge}
	\EndIf
	\For {$k = 1 \to n$}
		\State $S = [k]$
		\State {$\hat{x}_e \gets \binsearchh(k_e, e, \sum_{i \in S} M^i, \delta)$}
		\State {$\hat{M}_e \gets |S| l_e(\hat{x}_e) + \hat{x}_e l_e'(\hat{x}_e)$}
		\State {$(\hat{M}_e^i)_{i \le k} \gets \redistribh\left(\sum_{i \in S} M^i, \hat{M}, k, (M^i)_{i \le k}\right)$}
		\State $f_e^i = \frac{\hat{M}_e^i - l_e(\hat{x}_e)}{l_e'(\hat{x}_e)}$ for each player $i \in S$  \Comment{Note that $\sum_{i \in S} f_e^i = \hat{x}_e$}
		\If {($k=n$ \textbf{ or } $M^{k+1} \le l_e(\hat{x}_e) + 2n \Psi^2 \delta$)}
			\State $f_e \gets \hat{x}_e$, $S_e \gets S$
			\State {Continue with the next edge}
		\EndIf
	\EndFor
\EndFor
\State {$w_i \gets \sum_e f_e^i$ for each player $i$}
\State \Return {($\vec{f}, \vec{w}$)}
\end{algorithmic}
\end{algorithm}

Note that the time complexity of \graphflowh is $O\left(mn^2 \log (2n\Psi^3/\delta) \right)$, since it makes at most $mn$ calls to \binsearchh and \redistribh.

\begin{lemma}
Algorithm \graphflowh($\vec{M}, \delta$) returns flow vector $\vec{f}$ and demands $\vec{w}$ so that  $\vec{f}$ is an $4n \Psi^2 \delta$-equilibrium flow for demands $\vec{w}$, and $|w_i - \graphval{i}(\vec{M})| \le 4mn \delta \Psi^5$.
\label{lem:graphflowh}
\end{lemma}

\begin{proof}
We first show that the flow $\vec{f}$ is an $\epsilon$-equilibrium for demands $\vec{w}$. For each edge $e$, note that the players in $S_e$ are the only players that can have positive flow on the edge. Fix an edge $e$. Then $\hat{x}_e$ is the approximate solution returned by \binsearchh to the polynomial equation

\begin{align}
|S_e| l_e(x) + x l_e'(x) = \sum_{i \in S_e} M^i \, .\label{eqn:binsearchh}
\end{align}

\noindent Let $x_e^*$ be the exact solution to~\eqref{eqn:binsearchh}. Then by Lemma~\ref{lem:binsearchh}, $|x_e^* - \hat{x}_e| \le \delta$. Further, a simple calculation then gives us that

\[
|\hat{M}_e - \sum_{i \in S_e} M^i| \le n \delta \Psi + \Psi^2 \delta \le 2n \delta \Psi^2 \, .
\]

\noindent By Lemma~\ref{lem:redistribh}, this then gives us that for each player $i \in S_e$, $|\hat{M}_e^i - M^i|$ $\le |\hat{M}_e - \sum_{i \in S_e} M^i|$ $\le 2n \delta \Psi^2$. For each player $i \not \in S_e$, $f_e^i = 0$, and $M^i$ $\le M^{|S_e|+1} \le l_e(\hat{x}_e) + 2n \Psi^2 \delta$.

This is sufficient to show that $\vec{f}$ is an approximate equilibrium, since the minimum marginal cost on any edge is at least $M^i - 2n \Psi^2 \delta$, and if $f_e^i > 0$ on any edge, then the marginal cost is $\hat{M}_e^i$ which is at most $M^i + 2n \Psi^2 \delta$. Hence $\vec{f}$ is an $4n \Psi^2 \delta$-equilibrium.

We now show that $\vec{w}$ is also close to the true value, i.e., for each player $i$, $|w_i - \graphval{i}(\vec{M})| \le 4m \delta \Psi^4$. Fix an edge $e$, and let $S_e^*$ be the set of players with positive flow on edge $e$ as defined by Algorithm GraphFlow. Further, let $x_e^*(S_e)$ be the solution to

\[
|S_e| l_e(x) + x l_e'(x) = \sum_{i \in S_e} M^i \, ,
\]

\noindent while $x_e^*(S_e^*)$ be the solution to

\[
|S_e^*| l_e(x) + x l_e'(x) = \sum_{i \in S_e^*} M^i \, .
\]

Let $\hat{x}_e(S_e)$ and $\hat{x}_e(S_e^*)$ be similarly defined as the values returned by $\binsearchh$ with error parameter $\delta$. Thus $x_e^*(S_e^*)$ is the flow on edge $e$ as determined by Algorithm GraphFlow, while $|\hat{x}_e(S_e) - x_e^*(S_e)| \le \delta$ and $|\hat{x}_e(S_e^*) - x_e^*(S_e^*)| \le \delta$ by Lemma~\ref{lem:binsearchh}.

We first show that $S_e \subseteq S_e^*$. Suppose for a contradiction that $S_e \supset S_e^*$. Then the $|S_e^*|+1$th player enters $S_e$, but not $S_e^*$. Thus $M^{|S_e^*|+1} > l_e(\hat{x}_e(S_e^*)) + 2n \Psi^2 \delta$, while $M^{|S_e^*|+1} \le l_e(x_e^*(S_e^*)$, which is a contradiction, since $|l_e(\hat{x}_e(S_e^*)) - l_e(\hat{x}_e(S_e^*))|$ is at most $\delta \Psi$.

Thus $S_e \subseteq S_e^*$. Consider first the case that $S_e \subset S_e^*$. Then for the $|S_e|+1$th player, since it must have nonnegative flow as returned by GraphFlow, and because it does not belong to $S_e$,

\[
l_e(x_e^*(S_e^*)) \le M^{|S_e^*|+1} \le l_e(\hat{x}_e(S_e)) + 2n \Psi^2 \delta
\]

\noindent and hence, $|x_e^*(S_e^*) - \hat{x}_e(S_e)| \le 2n \Psi^3 \delta$. We note that $x_e^*(S_e^*)$ is the total flow returned by GraphFlow on edge $e$, while $\hat{x}_e(S_e)$ is the total flow on edge $e$ returned by \graphflowh. Then for any player $i$, the difference in flows on edge $e$ returned by GraphFlow and \graphflowh is

\[
\frac{M^i - l_e(x_e^*(S_e^*))}{l_e'(x_e^*(S_e^*))} - \frac{M^i - l_e(\hat{x}_e(S_e))}{l_e'(\hat{x}_e(S_e))} \ge \frac{M^i - l_e(\hat{x}_e(S_e)) + 2n\delta \Psi^2}{l_e'(\hat{x}_e(S_e)) - 2n\delta \Psi^2} - \frac{M^i - l_e(\hat{x}_e(S_e))}{l_e'(\hat{x}_e(S_e))} \le 4 n \delta \Psi^5
\]

Similarly, if $S_e = S_e^*$, then again $|x_e^*(S_e^*) - \hat{x}_e(S_e)| \le 2n \Psi^3 \delta$, and the above bound holds for the flow of any player on edge $e$. Hence, since the total change of any player's flow on an edge is at most $4 n\delta \Psi^5$, we get that for any player $i$, $|w_i - \graphval{i}(\vec{M})| \le 4m n\delta \Psi^5$.
\end{proof}

Note that the algorithm could return a flow vector $\vec{f}$ that has some negative entries. However, we can correct this in the following way. Suppose $f_e^i < 0$. Then it can be shown, from the proof of Lemmas~\ref{lem:graphflowh} and Claim~\ref{claim:graphfloweq}, that $|f_e^i| \le 4n \Psi^5 \delta$. We set $f_e^i = 0$, and maintain the total flow on edge $e$ unchanged by reducing the flow of players on edge $e$ that have positive flow. This is possible since $\hat{x}_e \ge 0$, and does not change the flow of any player by much (at most $4n^2 \Psi^5 \delta$). We do this for all edges, setting to zero the flow of any player with negative flow on any edge and decreasing the flow of other players to maintain the total flow on each edge. We redefine the demands $w_i = \sum_e f_e^i$ for all players $i$, and note that the demand of any player also does not change by more than $4mn^2 \Psi^5 \delta$ in this process. The resulting flow is nonnegative, with total flow on every edge unchanged. Further, it can be shown that $\vec{f}$ is an $8n^2 \Psi^6 \delta$-equilibrium flow for demands $\vec{w}$, and for all players $i$, $|w_i - \graphval{i}(\vec{M})| \le 8mn^2 \Psi^6 \delta$.

\begin{lemma}
Given a player $k$, and two marginal cost vectors $\vec{M}$ and $\vec{M}'$ that satisfy the following properties:
\begin{enumerate}
\item for all players $i < k$, $M^i = {M^i}'$,
\item for all players $i > k$, $|\graphvalh{i}(\vec{M},\delta) - \graphval{i}(\vec{M}')| \le \epsilon_i$.
\end{enumerate}
Let $\hat{w}_k = \graphvalh{k}(\vec{M},\delta)$, and $w_k' = \graphval{k}(\vec{M}')$. If $\hat{w}_k > w_k' + 4mn^2\delta\Psi^5 + \sum_{i > k} \epsilon_i$, then $M^k > {M^k}'$.
\label{lem:monotonenplayers2}
\end{lemma}

\begin{proof}
Let $M^k \le {M^k}'$, and let $P$ be the set of players $\{i \ge k: M^i \le {M^i}'\}$. Thus $k \in P$. We use the following property shown earlier in Lemma~\ref{lem:monotonenplayers}:

\begin{align}
\sum_{i \in P} \graphval{i}(\vec{M}) ~ \le ~ \sum_{i \in P} \graphval{i}(\vec{M}') \, . \label{eqn:Pincreases3}
\end{align}

\noindent The proof of this statement is exactly the same as in the earlier lemma. Further, we have the following properties from Lemma~\ref{lem:graphflowh}:

\[
\sum_{i \in P} \graphvalh{i}(\vec{M},\delta) ~ \le ~ \sum_{i \in P} \graphval{i}(\vec{M})  + 4mn^2 \delta \Psi^5\, .
\]

\noindent Together with~\eqref{eqn:Pincreases3}, this gives us that

\[
\hat{w}_k + \sum_{i \in P, i \neq k} \graphvalh{i}(\vec{M},\delta) ~ \le ~ w_k' + \sum_{i \in P, i \neq k} \graphval{i}(\vec{M})  + 4mn^2 \delta \Psi^5 \, .
\]

\noindent Since $|\graphvalh{i}(\vec{M},\delta) - \graphval{i}(\vec{M}')| \le \epsilon_i$ for all players $i > k$, this gives us that $\hat{w}_k \le w_k' + 4mn^2 \delta \Psi^5+ \sum_{i > k} \epsilon_i$ as required.
\end{proof}

\begin{corollary}
Given a player $k$, and two marginal cost vectors $\vec{M}$ and $\vec{M}'$ that satisfy the following properties:
\begin{enumerate}
\item for all players $i < k$, $M^i = {M^i}'$,
\item for all players $i > k$, $|\graphvalh{i}(\vec{M},\delta) - \graphval{i}(\vec{M}')| \le \epsilon_i$.
\end{enumerate}
Let $\hat{w}_k = \graphvalh{k}(\vec{M},\delta)$, and $w_k' = \graphval{k}(\vec{M}')$. If $M^k \le {M^k}' + \delta_1$, then $\hat{w}_k \le w_k' + 6mn^2(\delta+\delta_1) \Psi^4 + \sum_{i > k} \epsilon_i$.
\label{cor:monotonenplayers2}
\end{corollary}

\begin{proof}
Consider the marginal cost vector $\vec{N} := ((M^i)_{i \le k}, ({M^i}')_{i > k})$, obtained by replacing the $k$th component of $\vec{M}'$ by $M^k$. Then since $|M^k - {M^k}'| \le \delta_1$, by Claim~\ref{claim:continuous}, for each player $i$, $|\graphval{i}(\vec{N}) - \graphval{i}(\vec{M}')|$ $\le 2mn \Psi \delta_1$. Further, for all players $i > k$,
\begin{align}
|\graphvalh{i}(\vec{M},\delta) - \graphval{i}(\vec{N})| \le \epsilon_i + 2mn \Psi \delta_1 \, .
\label{eqn:monotonenplayers2}
\end{align}

Now consider marginal cost vectors $\vec{N}$ and $\vec{M}$. The first $k$ components in both are equal. Using~\eqref{eqn:monotonenplayers2} and Lemma~\ref{lem:monotonenplayers} gives us that
\[
|\graphvalh{k}(\vec{M},\delta) - \graphval{k}(\vec{N})| \le \sum_{i > k} \epsilon_i + 2mn^2 \Psi \delta_1 + 4mn^2 \Psi^5 \delta \le \sum_{i > k} \epsilon_i + 4mn^2 \Psi^5 (\delta + \delta_1) \, .
\]
\noindent and hence, since $|\graphval{i}(\vec{N}) - \graphval{i}(\vec{M}')|$ $\le 2mn \Psi \delta_1$, for all $i$,
\[
|\graphvalh{k}(\vec{M},\delta) - \graphval{k}(\vec{M}')| \le \sum_{i > k} \epsilon_i + 6mn^2 \Psi^5 (\delta + \delta_1) \, .
\]
\end{proof}

\begin{algorithm}[!ht]
\caption{\eqmcosth{k}($(M^1, \dots, M^{k-1}),\delta$)}\label{algo:eqmcost2}
\begin{algorithmic}[1]
\Require{Vector $(M^1, \dots, M^{k-1})$, with each component $M^i \in [0,\Lambda]$, error parameter $\delta$} \Comment{If $i=1$, there is no input required.}
\Ensure{Vector $(\vec{M})$ of marginal costs so that the first $k-1$ marginal costs are equal to the inputs, and for players $i \ge k$, the demand $\graphval{i}(\vec{M}) = v_i$.}
\If{$k = n$}
	\State {Using binary search in $[0,\Lambda]$, find $M$ so that $\graphval{n}((M^i)_{i < n}, M, \delta) \in [v_n \pm \delta]$.}
	\State \Return {$M$}
\EndIf
\State {$\low \gets 0$, $\high \gets \Lambda$}
\State {$\med \gets (\low + \high)/2$} \label{line:finalloop2}
\State {$(M^{k+1}, \dots, M^n) \gets \eqmcosth{k+1}((M^1, \dots, M^{k-1}, \med),\delta)$} \Comment{Call \eqmcosth{k+1} to get marginal costs for the remaining player $k+1, \dots, n$ so that the demand for these players is correct}
\If {$(\graphvalh{k}((M^i)_{i < k}, \med, (M^i)_{i > k},\delta) \in [v_k \pm 2^{(n-k)}6 mn^2 \Psi^5(\delta+\delta_1)])$ \textbf{or} $(High - Low \le \delta/(2m\Psi))$}
		\State \Return {$(\med, (M^i)_{i > k})$}
\ElsIf {$(\graphvalh{k}((M^i)_{i < k}, \med, (M^i)_{i > k}) > v_k + 2^{(n-k)}6 mn^2 \Psi^5(\delta+\delta_1)$}
		\State{$\high \gets \med$, goto~\ref{line:finalloop2}}
\ElsIf {$(\graphvalh{k}((M^i)_{i < k}, \med, (M^i)_{i > k}) < v_k - 2^{(n-k)}6 mn^2 \Psi^5(\delta+\delta_1)$}
		\State{$\low \gets \med$, goto~\ref{line:finalloop2}}
\EndIf
\end{algorithmic}
\end{algorithm}

\begin{lemma}
For any vector $(M^1, \dots, M^{k-1})$ with each component in $[0, \Lambda]$, the function $\eqmcosth{k}((M^1,$ $ \dots, M^{k-1}))$ returns marginal costs $(M^k, \dots, M^n)$ for the remaining players so that, for each player $i \ge k$, $\graphvalh{i}(M^1, \dots, M^n)$ $\in [v_i \pm 2^{(n-k)}6 mn^2 \Psi^5(\delta+\delta_1)]$. The time taken is \[O\left(mn^2 \log (2n \Psi^3/\delta) \left(\log (2 \Psi^2/\delta_1)\right)^{n-k+1} \right).\]
\label{lem:eqmcosth}

\begin{proof}
We will use the notation $N^{k-1} := (M^1, \dots, M^{k-1})$ for the input vector of marginal costs for the players $1, \dots, k-1$. The proof is by induction on $n$. In the base case, $k = n$, and the input is the vector $N^{n-1}$ $= (M^1, \dots, M^{n-1})$ with each component in $[0, \Lambda]$. By Claim~\ref{claim:existence1player}, there exists $M^*$ so that $\graphval{n}(N^{n-1}, M^*) = v_n$. Now consider the interval $[M^* \pm \delta_1]$. For any $\hat{M}$ in this interval, it follows from Claim~\ref{claim:continuous} that $\graphval{n}(N^{n-1}, \hat{M}) \in [v_n \pm 2mn\Psi \delta_1]$. Further, from Lemma~\ref{lem:graphflowh}, $\graphvalh{n}(N^{n-1}, \hat{M}, \delta)$ $\in [v_n \pm 2mn\Psi \delta_1 + 4m\delta \Psi^4]$. The contrapositive also holds true: if $\graphvalh{n}(N^{n-1}, $ $\hat{M}, \delta)$ $> v_n + 2mn\Psi \delta_1 + 4mn\delta \Psi^5$, then $\hat{M} > M^*$. Thus, $M^*$ must always be in the binary search interval, and the search must terminate if the length of the search interval is at most $\delta_1$ with $\hat{M}$ that satisfies $\graphvalh{n}(N^{n-1}, \hat{M}, \delta)$ $\in [v_n \pm 6mn^2\Psi^5 (\delta_1 + \delta)]$.

For the inductive step, we are given player $k < n$. We assume that given any input vector $N^k$ $=(M^1, \dots, M^k)$ with each component in $[0, \Lambda]$, $\eqmcosth{k+1}$ returns marginal costs $(M^{k+1}, \dots, M^n)$ for the remaining players so that for each of these remaining players $i \ge k+1$, $\graphvalh{i}(M^1, \dots, M^n)$ $\in [v_i \pm 2^{n-i}6mn^2\Psi^5 (\delta_1 + \delta)]$. We need to show that given any input marginal costs $N^{k-1} = (M^1, \dots, M^{k-1})$ for the first $k-1$ players, $\eqmcosth{k}$ finds marginal costs $(M^k, \dots, M^n)$ for players $k$ onwards so that the demand returned for these players $i \ge k$ by $\graphvalh{i}$ is in the interval $[v_i \pm 2^{n-i}6mn^2\Psi^4 (\delta_1 + \delta)]$.

Firstly, by Lemma~\ref{lem:existencenplayers}, choosing $S = [k, \dots, n]$ and $\hat{w}_i = v_i$ for players $i \in S$, there exist marginal costs $({M^*}^k, \dots, {M^*}^n)$ so that for all $i \ge k$, $\graphval{i}(N^{k-1}, ({M^*}^i)_{i \ge k}) = v_i$. We now show that the binary search procedure in $\eqmcosth{k}$ finds the required marginal cost $\hat{M}^k$. By the lemma, ${M^*}^k$ lies in the initial search interval $[0, \Lambda]$. Assume that in some iteration, ${M^*}^k$ lies in the search interval $[\low, \high]$, and $\med = (\low + \high)/2$. By the induction hypothesis, $\eqmcosth{k+1}(N^{k-1}, \med)$ returns marginal costs $(M^{k+1}, \dots, M^n)$ for the players $k+1, \dots, n$ so that for each of these players $i \ge k+1$ (but not player $k$), $\graphvalh{i}(N^{k-1}, \med, (M^i)_{i>k})$ $\in [v_i \pm 2^{n-i}6mn^2\Psi^5 (\delta_1 + \delta)]$. Further, for each player $i \ge k$, $\graphval{i}(N^{k-1}, ({M^*}^i)_{i\ge k}) = v_i$. Suppose that for player $k$,
\begin{align*}
\graphval{k}(N^{k-1}, \med, (M^i)_{i>k}) &> v_k + 2^{n-k}6mn^2\Psi^5 (\delta_1 + \delta)\\
 &= v_k + 6mn^2 \Psi^5 (\delta_1 + \delta) \left(1 + \sum_{i > k} 2^{n-i}\right) \,.
\end{align*}

\noindent Then by Lemma~\ref{lem:monotonenplayers2}, $\med > {M^*}^k$. In which case, the search procedure reduces the interval to $[\med, \high]$. Hence ${M^*}^k$ is always within the search interval. Further, from Corollary~\ref{cor:monotonenplayers2} and the induction hypothesis, if $\high - \low \le \delta_1$, then $\eqmcosth{k}(N^{k-1})$ returns marginal costs $(M^k, \dots, M^n)$ for the remaining players so that, for each player $i \ge k$, $\graphvalh{i}(M^1, \dots, M^n)$ $\in [v_i \pm 2^{(n-k)}6 mn^2 \Psi^5(\delta+\delta_1)]$.
\end{proof}
\end{lemma}

Let $\vec{f}$ be an $\epsilon$-equilibrium flow for ASRG $\varGamma$ with demands $\vec{w}$ for the players, and let $\vec{v}$ be a different vector of demands close to $\vec{v}$ so that $\Vert \vec{v} - \vec{w} \Vert_\infty \le \nu$. We now describe a simple procedure to modify $\vec{f}$ to obtain $\vec{g}$, which is an $\epsilon + 6 \nu \Psi^2$-equilibrium for demands $\vec{v}$. For each $i$, if $v_i \ge w_i$, we choose the edge on which player $i$ has minimum marginal cost, and increase player $i$'s flow so that $g_e^i = f_e^i + (v_i - w_i)$. If $v_i \le w_i$, on each edge $e$, we reduce $f_e^i$ on all edges while maintaining nonnegativity, until the total flow from player equals $v_i$. This gives us flow $\vec{g}$. We note that (1) if $g_e^i > 0$, then in flow $f$, $L_e^i(f) \ge \min_{e'} L_{e'}^i(f) - \epsilon$, (2) for any player $i$ and edge $e$, $|f_e^i - g_e^i| \le \nu$, and (3) for any edge $e$, the total flow changes by at most $n \nu$. The latter two properties are enough to establish that the change in marginal cost for any player on an edge $e$ $|L_e^i(g) - L_e^i(f L_e^i(\vec{f})|$ $\le 3m \nu \Psi^2$, and combined with the first property, this gives us that $\vec{g}$ is an $\epsilon + 6m \nu \Psi^2$-equilibrium for flow values $\vec{v}$.

From Lemma~\ref{lem:eqmcosth} and the above redistribution procedure, substituting $\nu = 4m \delta \Psi^4$, we get the required theorem.

\begin{theorem}
Algorithm \eqmcosth{1}($\delta$), combined with the redistribution procedure described to change flow values from $\vec{w}$ to $\vec{v}$, terminates in time \[O\left(mn^2 \log (2n \Psi^3/\delta) \left(\log (2 \Psi^2/\delta_1)\right)^{n-k+1} \right),\] and returns a $2^{n+5} m^2 n^2 \Psi^6 \delta$-approximate equilibrium.
\label{thm:implement}
\end{theorem}

Thus, an $\epsilon$-equilibrium can be obtained by this algorithm in time $O\left(mn^2 \left(\log (n\Psi/\epsilon) \right)^n\right)$.

\subsection{A game with multiple equilibria}
\label{sec:threeeq}

We first describe the construction of a two-player game on a six-edge network with nondecreasing, nonnegative, convex, and continuously differentiable cost functions on the edges that has exactly three equilibria, two of which are rational and one of which is irrational. The game has two players $b$ and $r$ with demands $v_b=4763.5$ and $v_r=2415.3$ respectively. The costs for the edges $e_1$ and $e_5$, and for $e_2$ and $e_4$ are same. The detailed cost functions are given in the Table \ref{tab:mul_eq_latencies}.

\begin{figure}[ht] 
\centering
\includegraphics[scale=.8]{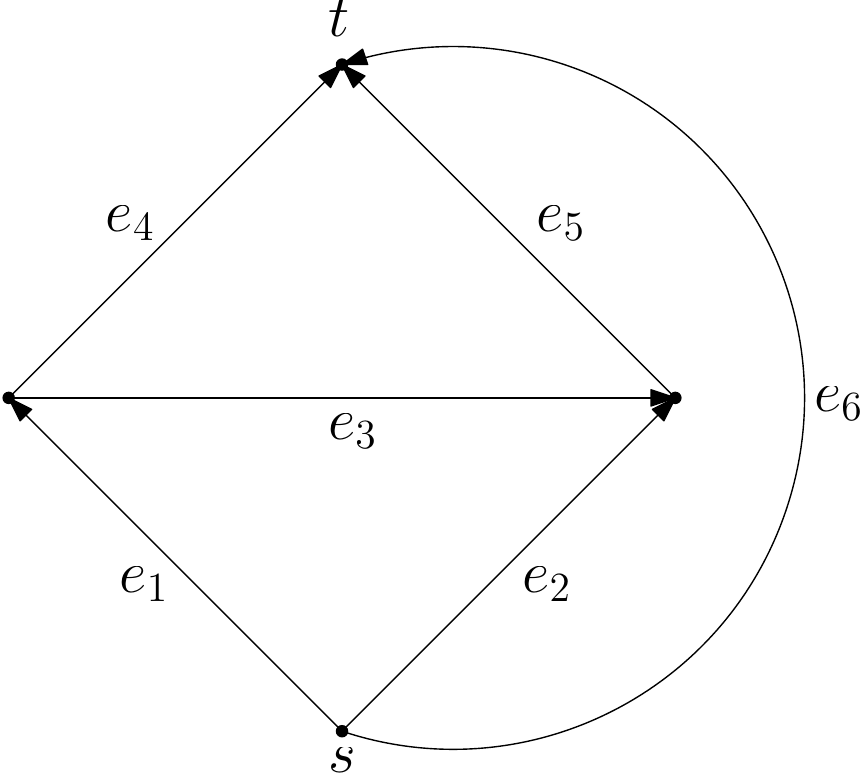}
\caption{ASRG $\mathcal{T}$, with multiple equilibria}
 \label{fig:mul_eq}
\end{figure}

\begin{table}[!ht]  
  \caption{cost functions of ASRG $\mathcal{T}$ with multiple equilibria}
  \centering
    \begin{tabular}{|l| l|}
   \hline
    Edge & Delay Function \\
     \hline
$e_1,e_5$  &
$\begin{aligned}
\begin{cases}
    0.1694x + 293.1103 & x \leq 599 \\
    \frac{329219 x^2 - 394202776 x +  118472907676}{1190000} & 599 <x <599+\frac{119}{173} \\
    0.55x+ 65 & x \geq 599+\frac{119}{173}
  \end{cases}
 \end{aligned}
$                                             \\
 \hline
   $e_2$ , $e_4$ & $0.02x + 670$ \\
   \hline
   $e_3$ & $0.06x +208$\\
   \hline
   $e_6$ &
    $\begin{aligned}
   \begin{cases}
     \frac{36588331}{303185000}x + \frac{993797009}{606370000} & x \leq \frac{278245}{46}\\
     \frac{11040x-18059920}{66617} & x \geq \frac{11028504038037611}{1819515107546}\\

     \alpha x^2 - \beta x + \gamma & \text{otherwise}

 \end{cases}
   \end{aligned}$\\
   \hline
    \end{tabular}\label{tab:mul_eq_latencies}
  \end{table}

\bgroup
\def\arraystretch{1.4}
\begin{table}[!ht]
\centering
\begin{tabular}{|l|l|}
\hline
 $ \alpha$ &  $\frac{827658806647032986535529}{456723902937873797726820000} $\\
 \hline
$ \beta $ &  $ \frac{9957574988143115788690556303}{456723902937873797726820000} $ \\
\hline
$\gamma$ &  $\frac{121132628509821789002909277787321}{1826895611751495190907280000} $ \\
\hline
\end{tabular}
\end{table}
\egroup

We now obtain some properties regarding equilibria in this ASRG, culminating in showing that there are the three required equilibria. We will use the following notation to denote the paths in game $\mathcal{T}$:
\begin{align*}
 P_1 & \text{ is path } e_1-e_4\\
 P_2 & \text{ is path } e_2-e_5\\
 P_3 & \text{ is path } e_1 -e_3 -e_5\\
 P_4 & \text{ is path } e_6.
\end{align*}

We first show that at equilibrium, for each player, the flow on edges $e_1$ and $e_5$ are equal, as are the flow on edges $e_2$ and $e_4$. Note that the cost functions on these edges is also the same. We also show that at any equilibrium, both players have non-zero flow on path $P_4$ consisting of edge 6.
\begin{lemma} \label{lem:mul_eq_flowtypes}
  If $f$ is an equilibrium flow of ASRG $\mathcal{T}$ shown in Figure $\ref{fig:mul_eq}$, then for each player $i \in \{b,r\}$, (1) $f_1^i = f_5^i$, (2) $f_2^i = f_4^i$, and (3) $f_6^i > 0$.
\begin{proof}
(1) and (2): Assume contradiction. Suppose $f_1\geq f_5$ and $f_1^b>f_5^b\geq 0$. Observe that, this implies $f_4 > f_2$ and $f_{4}^{b}>f_{2}^{b}\geq 0$. Therefore, $L_{P_1}^{b}>L_{P_2}^b$. Since, flow of player $b$ on path $P_1$($f_1^b >0$ and $f_4^b> 0$), we get $L_{P_1}^{b}\leq L_{P_2}^b$ which is a contradiction. The other cases can be argued similarly. This proves first two parts of the Lemma.
 \noindent
 (3): Suppose $f^b_{6} = 0 $. Therefore, marginal cost of player $b$ on edge $e_6$ is
 \begin{align*}
 L_{P_4}^{b} (f) & = l_6(f_6) \\
 & \leq  l_6(v_r)\\
 &\approx 293.12.
 \end{align*}
 It can be seen from Table \ref{tab:mul_eq_latencies} that the cost on edges $e_1=e_5$, $e_2=e_4$ and $e_3$ with zero flow (i.e., the constant term for the cost functions) is $293.1103,\ 670$ and $208$ respectively. Hence, the marginal costs of any player on paths $P_1$, $P_2$ and $P_3$ are $\approx 963.1103$, $\approx 963.1103$ and $\approx 2684.93$ respectively. All of these are much larger than $L_{P_1}^b(f)$, hence player $b$ cannot be using any of these paths either giving us a contradiction. The case for player $r$ follows similarly.
\end{proof}
 \end{lemma}

By the lemma, given the flow values of players $b$ and $r$ on edges $e_1$ and $e_2$, then the flows on other edges can be calculated easily since $v_b$ and $v_r$ are known. Therefore, from now on we will consider the flows as shown in Figure \ref{fig:mul_eq_flows}. The variables $x_1$, $x_2$ correspond to the flow of player $b$ and $y_1$, $y_2$ correspond to the flow of player $r$. Observe that $x_2\geq x_1$ and $y_{2}\geq y_{1}$, since the flow on edge $e_3$ is non-negative.

\begin{figure}[ht]
\centering
\includegraphics[scale=.8]{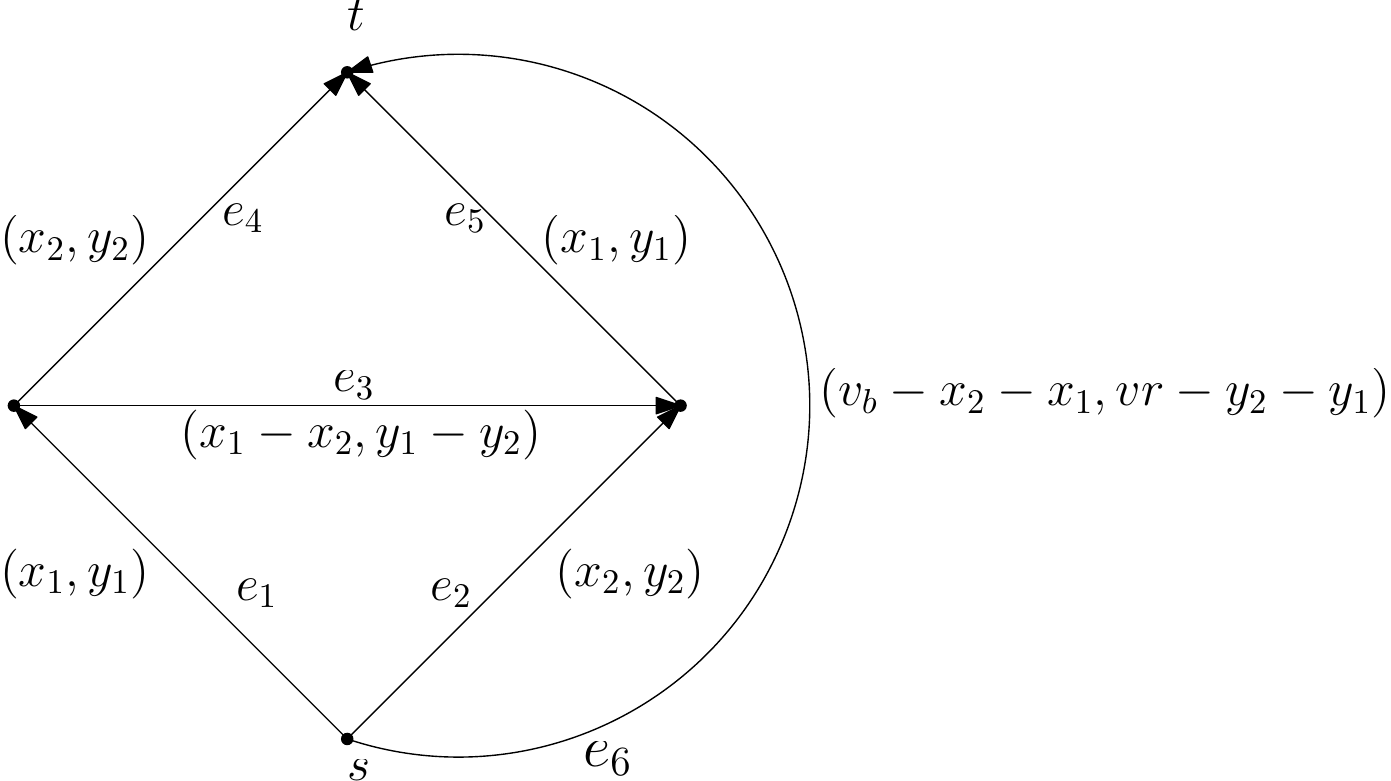}
\caption{Flows in $\mathcal{T}$.}
 \label{fig:mul_eq_flows}
\end{figure}

 \begin{lemma}\label{lem:cost_is_algebraic}
 The ASRG $\mathcal{T}$ has three equilibria $f$, $g$, and $h$ such that $f$ and $g$ are rational equilibria, while $h$ is an irrational equilibria. That is, in $h$, there exists an edge $e$ and a player $i$ with $h_e^i$ irrational.

\begin{proof}
For the proof of the lemma, we will first argue that every equilibrium is a common zero of a set of polynomials arising from the equilibrium conditions of players. Then, we enumerate over all the possible zeros to check if they satisfy all the equilibrium conditions (this is done by Mathematica).

By Lemma~\ref{lem:mul_eq_flowtypes}, we know that $f_{P_4}^b>0$ and $f_{P_4}^r>0$. Combining this with flow conservation as shown in Figure~\ref{fig:mul_eq_flows}, we get the following conditions at any equilibrium:

\begin{align}
 L_{P_4}^b(f) &\leq L_{P_1}^b(f), &&   x_2\geq 0, \label{eq_cond1}\\
 L_{P_4}^b(f) &\leq L_{P_3}^b(f), &&   x_1\geq x_2,\label{eq_cond2}\\
 L_{P_4}^r(f) &\leq L_{P_1}^r(f), &&   y_2\geq 0,\label{eq_cond3}\\
 L_{P_4}^r(f) &\leq L_{P_3}^r(f),  &&   y_1\geq y_2.\label{eq_cond4}
 \end{align}\label{eq_cond}

Observe that at equilibrium, for each of (\ref{eq_cond1}), (\ref{eq_cond2}),(\ref{eq_cond3}) and (\ref{eq_cond4}), at least one of the two inequalities must be tight. For example if $x_1>0$ and $x_2=0$, then $L_{P_4}^b(f) \leq L_{P_1}^b(f)$ and $ L_{P_4}^b(f) = L_{P_3}^b(f)$. Therefore, at any equilibrium flow, for each player we have two equalities and two inequalities depending on the values of $x_1$, $x_2$, $y_1$ and $y_2$. Since our cost functions are piece-wise polynomials of degree $\leq 3$, with $x_1$, $x_2$, $y_1$ and $y_2$ as variables, the equalities are polynomial equations with degree $\leq 3$.

Therefore, values of $x_1,x_2,y_1$ and $y_2$ such that the flow is an equilibrium of $\mathcal{T}$, are subset of common zeros of polynomial equations with degree $\leq 3$. We will enumerate over all the possible common zeros of polynomial equations and show that there are $3$ equilibrium flows for ASRG $\mathcal{T}$ among with one is an irrational equilibrium (with irrational amount of flow on edges) and other two are rational equilibrium.

In order to check for equilibria, we consider $9 \times 16 = 144$ possible states for the flows. The states are obtained as follows. Firstly, the cost functions for edges $e_1$, $e_5$ and $e_6$ are made up of three piecewise functions. Since the flow on $e_1$ and $e_5$ is always equal by Lemma~\ref{lem:mul_eq_flowtypes}, we get 9 states, depending on which piece of the cost function is used in equilibrium calculation. The notation state $(i,j)$ is used for the state when calculations are using the $i$th piece of cost functions on edges 1 and 5, and $j$th piece of edge on edge 6.

Further, for each state, each of the two players could use any subset (including the empty set) of the set of paths $\{P_1, P_3\}$. Note that by Lemma~\ref{lem:mul_eq_flowtypes}, both players must use path $P_4$ at equilibrium, and they use $P_2$ iff they use $P_1$. Thus each player has four options, giving us 16 options for the equilibrium flow within each state. This gives us a total of 144 possible states to check for equilibria.

Each state gives us a different set of polynomial equalities and inequalities. For each state, we obtain the common zeros of polynomials by finding Groebner basis\footnote{A very brief introduction to Gr\"{o}bner bases is given in Section~\ref{sec:grobner}.}, which allows us to more easily calculate the set of common zeros by solving the polynomials one after another independently, in a manner similar to Gaussian elimination. We then verify that the solution thus obtained to the polynomial equalities satisfies the inequalities, and also that the flow is correct for the particular state it is computed in. More details are provied in Sections \ref{sec:grobner} and \ref{sec:verifyingeq}. By checking over all the possible states, we get that there are three possible equilibria:

\begin{itemize}
\item Flow $g$ for State (1,1) - case $x_2=x_1>0, y_1>0, y_2=0$,
\item Flow $f$ for State (3,3) - case $x_2=x_1>0, y_1>0, y_2=0$,
\item Irrational Equilibrium flow $h$ for State(2,2) - case $x_2=x_1>0, y_1>0, y_2=0$.
\end{itemize}

The exact values of $x_1$, $x_2$, $y_1$ and $y_2$ for these equilibria are given in the Mathematica file ``Flows.nb". In Section~\ref{sec:verifyingeq}, we describe in detail how to verify that these three are the only equilibria using Mathematica. Source files are available on author's \href{www.tifr.res.in/~phaniraj}{webpage}.
\end{proof}
\end{lemma}

\begin{figure}[ht]
\centering
\includegraphics[scale=0.75]{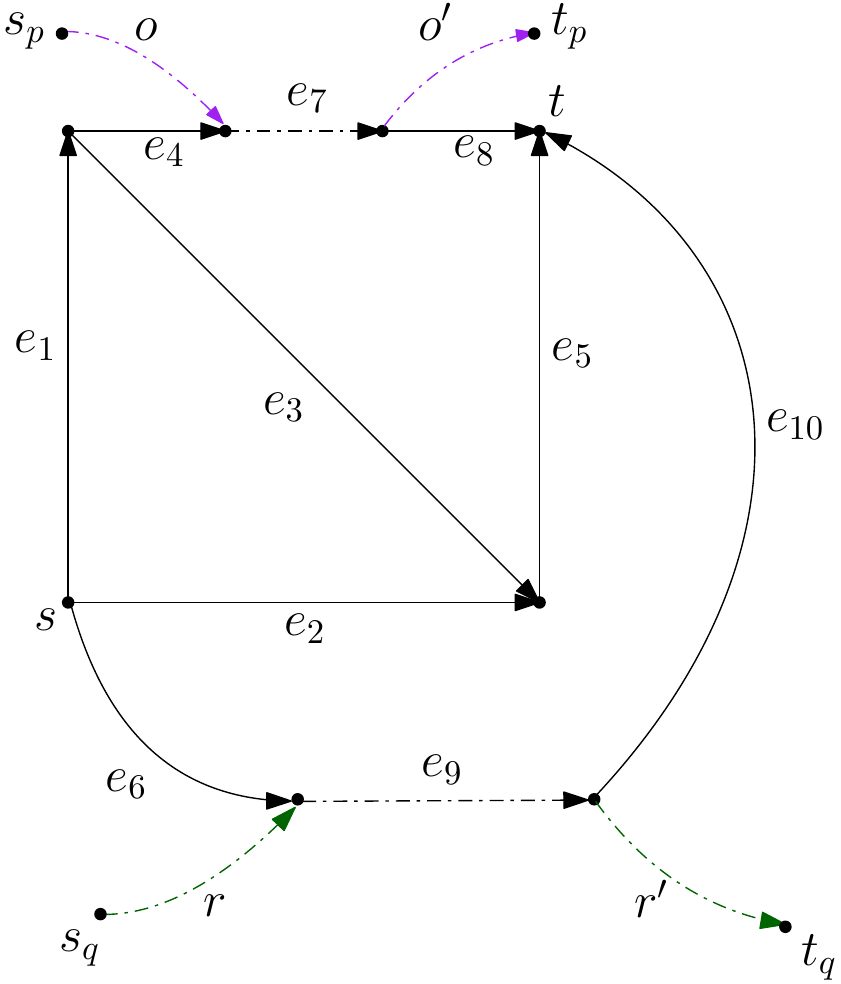}
\caption{ASRG $\mathcal{G}$, with multiple equilibria}
 \label{fig:mul_eq_game_3players2}
\end{figure}

We now modify ASRG $\mcT$ slightly to obtain $\mcG$ (Figure~\ref{fig:mul_eq_game_3players2}), which will be used next in the construction of the ASRG for the reduction. We will show in Lemma~\ref{lem:mul_eq_T} that the properties of equilibria in $\mcG$ and $\mcT$ are almost identical.

For ASRG $\mcG$, the cost functions on edges 1, 2, 3, 5 are identical to $\mcT$. The cost function on edge 4 of $\mcT$ is however one more than the sum of the cost functions on edges 4, 7, and 8 in $\mcG$, as follows:

\[  l_{e_4}(x) = 0.01x + 669.5, \quad l_{e_7}(x) = l_{e_8}(x) = 0.005x, \, \]
  \begin{align*}
  l_{e_6}(x) &=\begin{cases}
    \frac{36588331}{303185000}x + \frac{993797009}{606370000} - \theta x-50\theta & x \leq \frac{278245}{46}\\
    \frac{11040x-18059920}{66617} - \theta x-50\theta  & x \geq \frac{11028504038037611}{1819515107546}\\
    \alpha x^2 - \beta x + \gamma - \theta x - 50\theta & \text{otherwise}
\end{cases}\\
\text{where\ \ } \theta = 4/300.\\
l_{e_9}(x)&= 2x/300,\\
l_{e_{10}}(x)&= 2x/300.
\end{align*}

where $\alpha,\beta,\gamma$ are same as in Table~\ref{tab:mul_eq_latencies}.
\noindent This describes the cost functions on the edges of $\mcG$. There are now 4 players, $b$, $r$, $p$ and $q$. Again, for players $b$ and $r$, the demands are the same as earlier: $v_b = 4763.5$ and $v_r = 2415.3$. Players $p$ and $q$ have demand of $100$ units and by construction has can send their demand through a single path i.e., $s_p$-$t_p$ path and $s_q-t_q$ respectively.

\begin{lemma}\label{lem:mul_eq_T}
The ASRG $\mathcal{G}$ shown in Figure \ref{fig:mul_eq_game_3players2} has $3$ equilibria $f$, $g$, and $h$ that are identical to the equilibrium flows in $\mcT$.
 \begin{proof}
 We will show the result for player $p$ and proof for $q$ is quite similar. Suppose $f$ is an equilibrium flow for ASRG $\mathcal{T}$ and $g$ is an equilibrium flow for ASRG $\mathcal{G}$.
 Before proceeding further, observe the following:
 \begin{itemize}
  \item  The network of ASRG $\mathcal{G}$ is very similar to that of $\mathcal{T}$ except $e_4$ in $\mcT$ is replaced by $e_4-e_7-e_8$ in $\mcG$ and $e_6$ in $\mcT$ is replaced by $e_6-e_9-e_{10}$ in $\mcG$.
  \item The cost functions on edges $e_1,e_2,e_3,e_5$ for $\mathcal{G}$ and $\mathcal{T}$ are identical.
  \item In ASRG $\mathcal{G}$, player $p$ sends $100$ units of flow from $s_p$ to $t_p$ via edges $o-e_7-o'$ and player $q$ sends $100$ units of flow from $s_q$ to $t_q$ via edges $r-e_9-r'$. Since players $p$ and $q$ can send their flow only via one path, their equilibrium conditions are always satisfied.
 \end{itemize}
The marginal cost of player $b$ on edge $e_4$ in ASRG $\mathcal{T}$ is
\begin{align*}
L_{e_{4}}(f) &= l_{e_4}(f)+f_{e_4}^b l_{e_4}'(f)
\\&= 0.02 (f_{e_4}^{b}+f_{e_4}^{r})+670+0.02 (f_{e_4}^{b})\\
&=0.04(f_{e_4}^{b}) +0.02 (f_{e_4}^{r}) + 670.
\end{align*}
Similarly marginal cost of player $b$ on sub path $e_7-e_4-e_8$ in $\mathcal{G}$ is
\begin{align*}
L_{e_{4}}(g)+L_{e_{7}}(g)+L_{e_{8}}(g) &=0.04(g_{e_4}^{b}) +0.02 (g_{e_4}^{r}) + 670.
\end{align*}

Therefore, as far as player $b$ is concerned his equilibrium conditions are same in $\mathcal{T}$ and $\mathcal{G}$, provided the flows are same on other edges. Therefore, his equilibrium conditions on any path in game $\mathcal{G}$ are same to that of $\mathcal{T}$. Similarly one can show that the same of player $r$.

By Lemma \ref{lem:cost_is_algebraic}, there are three solutions which satisfy the equilibrium conditions of $\mathcal{T}$. Therefore, there are three equilibria for ASRG $\mathcal{G}$ out of which two are rational equilibria and one is an irrational equilibrium.

Therefore, equilibrium flows for player $b$ and $r$ are same in both the games $\mathcal{G}$ and $\mathcal{T}$ except for the edges $e_4,e_7,e_8$. But, the flow values on these edges can be calculated easily once the flows on other edges are known. So, we will denote the equilibrium flows in both the games by $f,g$ and $h$ as far as the players $b$ and $r$ are concerned.
\end{proof}
\end{lemma}

The equilibrium flows(with approximate values) for the game $\mcG$ are given in Table \ref{tab:eq_flows}. The exact values are given in the file ``Flows.nb"  attached.

\begin{table}[!ht]  
\centering
\caption{Equilibrium flows for ASRG $\mcG$ accurate up to $2$ decimals}
\begin{tabular}{|c|c|c|c|c|c|c|}
\hline
\multirow{2}{*}{Edge} & \multicolumn{2}{c|}{$f$} &\multicolumn{2}{c|}{g} &\multicolumn{2}{c|}{h}\\
\cline{2-7}
& $b$ & $r$ &  $b$ & $r$ & $b$ & $r$
\\ \hline
  e1 & $500$ & $100$ &  $540$ & $50$ & $527.41$ & $71.78$ \\
  e2 & $500$ & $0$ & $540$ &$0$ & $527.41$ &$0$ \\
  e3 & $0$ & $100$ & $0$ & $50$ & $0$&$71.78$  \\
  e4 & $500$ & $0$ & $540$ & $0$ & $527.41$ &$0$ \\
  e5 & $500$ & $100$ & $540$ & $50$ & $527.41$ & $71.78$ \\
  e6 & $3763.5$ & $2315.30$ & $3683.5$ & $2365.30$ &  $3708.69$ &$2343.54$ \\
  e7 & $500$ & $0$ & $540$ & $0$ & $527.41$ &$0$ \\
  e8 & $500$ & $0$ & $540$ & $0$ & $527.41$ &$0$ \\
  e9 & $3763.5$ & $2315.30$ & $3683.5$ & $2365.30$ &  $3708.69$ &$2343.54$ \\
  e10 & $3763.5$ & $2315.30$ & $3683.5$ & $2365.30$ &  $3708.69$ &$2343.54$ \\

  \hline
\end{tabular}\label{tab:eq_flows}
\end{table}

The following claim will be later used to show that there exists an equilibrium where players $p$ and $q$ have cost at most $C$ iff the given set $S$ satisfies \subsum.

\begin{claim}\label{claim:cost_sum}
  For game $\mcG$ we have
    \[C_{e_7}^p(f)+C_{e_9}^q(f)=C_{e_7}^p(g)+C_{e_9}^q(g).\]
  \begin{proof}
    Let $f_{7}$ be the sum amount of flow on edge $e_7$ by players $b$ and $r$. Similarly $g_{7}$, $f_{9}$ and $g_{9}$ are defined. Now recall that $l_{e_7}(x)=0.005x$ and $l_{e_9}(x)=2x/300$ and $v_p=v_q=100$. Now observe that
    \begin{align*}
      C_{e_7}^p(g)-C_{e_7}^p(f) & = 100 \left\{l_{e_7}(100+g_{7})-l_{e_7}(100+f_{7}) \right\}\\
      &= 0.5 (40)\\
      &= 20.
    \end{align*}
  Similarly, it can be seen that $C_{e_9}^q(f)-C_{e_9}^q(g)=20$. Hence the claim.
  \end{proof}
  \end{claim}

\subsection{Groebner Basis}
\label{sec:grobner}

Consider the ring of polynomials of $n$ variables over field $K$, $K[x_1,x_2,\ldots,x_n].$ If $\mathcal{F}=\{p_1,p_2,\ldots,p_k\}$ is set of polynomials, then the ideal generated by $\mathcal{F}$ is
\[\langle \mathcal{F}\rangle=\left\{ p_1q_1+p_2q_2+\ldots+p_kq_k|q_1,q_2,\ldots,q_k\in K[x_1,x_2,\ldots,x_n] \right\}.\]

The variety of $\mathcal{F}$ is the set of all common zeros of all the polynomials in $\mathcal{F}$ i.e.,\[\mathcal{V(F)}=\{x_1,x_2,\ldots,x_n\}\in K^n|p(x_1,x_2,\ldots,x_n)=0,\forall p\in \mathcal{F}\}.\]

Given an ideal generated by polynomials $\langle \mathcal{F}\rangle$, Groebner basis for the $\langle \mathcal{F}\rangle$ is the set of polynomials $g_1,g_2,\ldots,g_t$ such that they generate $\langle \mathcal{F}\rangle$, with some desirable algorithmic properties. This is encapsulated by the following Lemma.

\begin{lemma}\cite{cox}
If $G$ is the Groebner basis of $\langle \mathcal{F}\rangle$, then
\[\mathcal{V(F)}=\mathcal{V(G)}.\]
\end{lemma}

Therefore, to find the common zeros of polynomials in $\mathcal{F}$, it is enough to find the Groebner basis of $\mathcal{F}$ and then finding the common zeros of the Groebner basis. One big advantage we have by finding the Groebner basis, is that one can easily calculate the set of common zeros of polynomials by solving the polynomials one after another independently. For example, if $p_1[x_1,x_2,\ldots,x_n]$, $p_2[x_1,x_2,\ldots,x_n]$, $\ldots$, $p_n[x_1,x_2,\ldots,x_n]$ are the polynomials, then the corresponding Groebner basis are the polynomials $g_1[x_1]$, $g_2[x_1,x_2]$,$\ldots$, $g_n[x_1,x_2,\ldots,x_n]$. Then one can solve $g_1[x_1]=0$ to get the roots of $g_1$. Then we will substitute each of the root for $x_1$ in $g_2[x_1,x_2]$ and then solve for the roots of the corresponding polynomials. We will continue this process till the end to get the set of all common zeros of the polynomials. We give an example of this in Section~\ref{sec:verifyingeq}.

\textbf{Remark:}
A very concise introduction to Groebner basis is given by Bernd Sturmfels ~\cite{bernd}. The method of solving the polynomials by reducing to Groebner basis can be found in Chapter $2$ of ~\cite{cox}.

\subsection{Verifying Equilibria}
\label{sec:verifyingeq}

Before going further, recall that the cost functions for edges $e_1,\ e_5$ and $e_6$ are made up of three piecewise functions. Therefore, depending on the amount of flow on each edge, the cost is calculated with the corresponding piecewise function. We will denote this using states $1,2,3$ as shown in the Figure \ref{fig:states}. By Lemma \ref{lem:mul_eq_flowtypes}, the flow on edges $e_1$ and $e_5$ is identical, hence there are 9 possible states. We will first give an example to check the equilibrium conditions for one of the possible cases. Later we will give an outline of how to use the Mathematica files provided to check over all the possible cases.
 \begin{figure}[ht]
\centering
\includegraphics[scale=1]{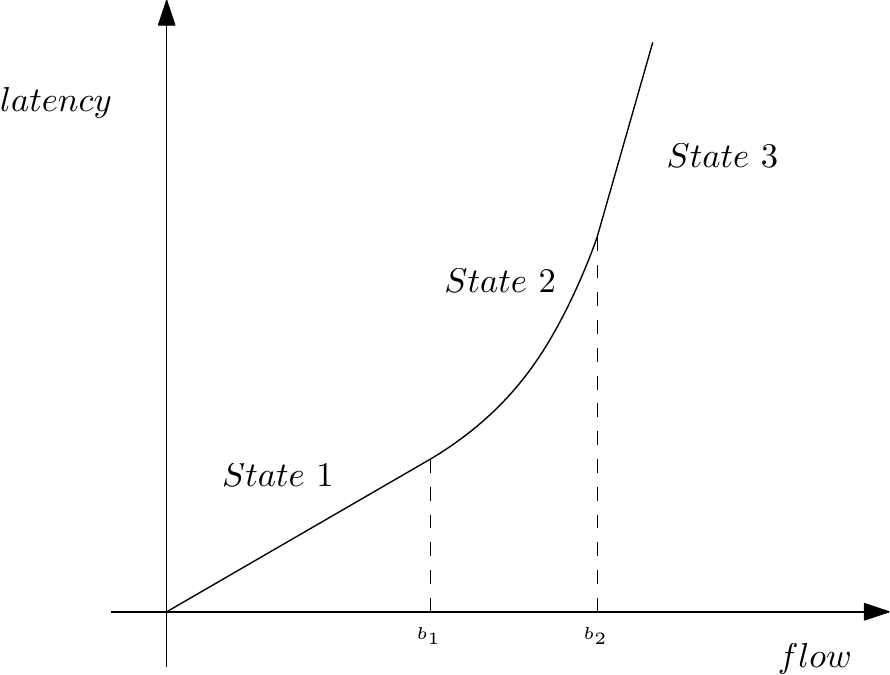}
\caption{Possible equilibrium flows for $\mathcal{T}$.}
 \label{fig:states}
\end{figure}

Consider the case for which the flow is such that both the edges $e_1$ and $e_6$ are in state $1$ which we will denote as $STATE (1,1)$ and $x_2=x_1>0, y_1>0, y_2=0$. The corresponding equilibrium conditions are

\begin{align*}
 130599901 x_1 + 43973935 y_1-72722643290=0\\
 x_2=x_1\\
 y_2=0\\
 87947870 x_1 + 157498509 y_1-55366775250 =0.
\end{align*}

The Groebner basis for the above set of polynomials is
\begin{align*}
 y_1-50=0\\
 x_2=x_1\\
 y_2=0\\
 x_1 -540=0.
\end{align*}

This corresponds to one of the possible equilibrium flows. Before going further, observe that for each possible $STATE (i,j)_{i,j\in \{1,2,3\}}$, there are 16 possible types of flows. This can be seen easily since $x_1 \geq x_2 \geq 0$, there are four possible types flows for player $b$ namely $x_1=x_2=0$, $x_1>x_2=0$, $x_1=x_2>0$ and $x_1>x_2>0$. Similarly there are $4$ possible types of flows for player $r$ and so potentially an equilibrium flow can be any of those 16 cases. Therefore, for each state, we will consider these 16 types of flows and check if the equilibrium is possible or not.

Now we will give step by step instructions  to verify each case. In the attached zip folder there are 9 files named $(``STATE\ i,j.nb)_{i,j\in \{1,2,3\}}$ corresponding to each possible state. In each such file we will check the existence in all the possible $16$ cases for the corresponding state.
\begin{table}[ht]
\centering
\caption{Description of Mathematica files}
\label{tab:math_files}
\begin{tabular}{|l|l|}
\cline{1-2}
\multicolumn{1}{|l|}{File} & Description                                                                                  \\ \cline{1-2}
Functions.nb               & Defines all the  cost functions along with other relevant data.                           \\
Flows.nb                   & Contains the values for the possible equilibria for game $\mathcal{T}$.                      \\
STATE $i,j$.nb for $i$, $j \le 3$               & Checks the feasiblity of equilibrium conditions for STATE(i,j)
\\ \cline{1-2}
\end{tabular}
\end{table}
Follow the following instructions to compile the Mathematica files.
\begin{enumerate}
\item Evaluate(run) ``Functions.nb". One can evaluate the Mathematica file by $Evaluation \rightarrow $ \\ $ Evaluate Notebook $.
 \item Evaluate the interested $(``STATE\ i,j.nb)_{i,j\in \{1,2,3\}}$ to check the possible cases.
\end{enumerate}

\textbf{Remark:} Before evaluating ``Functions.nb", the notebook should be changed to Global context by $Evaluation \rightarrow Notebook's\ Default\ Context \rightarrow Global`$. There are further comments in ``Functions.nb" and STATE 1,1.nb files.

Further explanation can be found in ``Functions.nb" and ``STATE 1,1.nb" files.

By checking over all the possible cases, we get that there are three possible equilibria corresponding
\begin{itemize}
\item Flow $g$ for State(1,1) - case $x_2=x_1>0, y_1>0, y_2=0$,
\item Flow $f$ for State(3,3) - case $x_2=x_1>0, y_1>0, y_2=0$,
\item Irrational Equilibrium flow $h$ for State(2,2) - case $x_2=x_1>0, y_1>0, y_2=0$.
\end{itemize}

The exact values of $x_1,x_2,y_1$ and $y_2$ for these equilibria are given in the Mathematica file ``Flows.nb".

\end{document}